\documentclass[12pt]{article}
\usepackage[utf8]{inputenc}
\usepackage[english]{babel}
 \usepackage{amsmath}
\usepackage{amsthm}
\usepackage{amssymb}
\usepackage{hyperref}
\usepackage{braket}
\usepackage{color}
\usepackage{comment}
\usepackage{graphicx}
 \usepackage{stmaryrd}
 \usepackage{xcolor}

\usepackage[margin=3cm]{geometry}

   \newcommand{\Ab}[1]{\llbracket #1 \rrbracket}
\renewcommand{\Gamma}{\mathbb{Z}^d}
\newcommand{\epsi}{\varepsilon}

 \newcommand{\R}{ \mathbb{R} }

\newcommand{\N}{ \mathbb{N} }
\newcommand{\Z}{ \mathbb{Z} }
\newcommand{\D}{\mathrm{d}}
\newcommand{\Or}{{\mathcal{O}}}

\newcommand{\bm}{\begin{pmatrix}}
\newcommand{\Em}{\end{pmatrix}}

 \newcommand{\Aloc}{\mathcal{A}_{\rm loc}}
  \newcommand{\SLT}{SLT}
 \newcommand{\Lambdak}{{\Lambda_k}} 
   \newcommand{\Lambdal}{{\Lambda_l}} 
     \newcommand{\LambdaM}{{\Lambda_M}} 
    \newcommand{\LambdaMp}{{\Lambda_{M'}}} 
        \newcommand{\LambdaN}{{\Lambda_N}}

\begin{document}
	\allowdisplaybreaks
\numberwithin{equation}{section}
 	\theoremstyle{definition}
	\newtheorem{defi}{Definition}[section]
\newtheorem{rmk}[defi]{Remark}

\theoremstyle{plain}
	\newtheorem{thm}[defi]{Theorem}
		\newtheorem{lem}[defi]{Lemma}
	\newtheorem{cor}[defi]{Corollary}
	\newtheorem{prop}[defi]{Proposition}

\title{Adiabatic theorem in the thermodynamic limit:\\ Systems with a uniform gap.} 
\author{Joscha Henheik \\
	\textit{\footnotesize Mathematisches Institut, Eberhard-Karls-Universit\"at,} \\{ \it \footnotesize Auf der Morgenstelle 10, 72076
		T\"ubingen, Germany} \\
	{\footnotesize and} \\
	{\it \footnotesize IST Austria, Am Campus 1, 3400 Klosterneuburg, Austria} \\[7mm]
	Stefan Teufel\footnote{stefan.teufel@uni-tuebingen.de} \\ \textit{\footnotesize Mathematisches Institut, Eberhard-Karls-Universit\"at,} \\{ \it \footnotesize Auf der Morgenstelle 10, 72076
		T\"ubingen, Germany} }
\maketitle

\begin{abstract}
We show that recent results on adiabatic theory for  interacting gapped many-body systems on finite lattices remain valid in the thermodynamic limit. More precisely, we prove a generalised super-adiabatic theorem  for the automorphism group  describing the   infinite volume dynamics on the quasi-local algebra of observables.
The key assumption is the existence of a sequence of {\em gapped} finite volume Hamiltonians which generates the same infinite volume dynamics in the thermodynamic limit.
Our adiabatic theorem holds also for certain perturbations of gapped ground states that close the spectral gap (so it is an adiabatic theorem also for resonances and in this sense `generalised'), and it provides an adiabatic approximation to all orders in the adiabatic parameter (a property often called `super-adiabatic'). In addition to existing results for finite lattices, we also  perform a resummation of the adiabatic expansion and allow for observables that are not strictly local. Finally, as an application, we prove the validity of linear and higher order  response theory for our class of perturbations also for   infinite systems. While we consider the result and its proof as new and interesting in itself, they also lay the foundation for the proof of an adiabatic theorem for systems with a gap only in the bulk, which will be presented in a follow-up article.
\end{abstract}
\newpage
\section{Introduction}	
 We consider time-dependent families 
\[
H^\Lambda(t) = H_0^\Lambda(t) +\epsi V^\Lambda(t)\,,\quad t\in I \subset \R\,,
\]
 of many-body  Hamiltonians for  lattice fermions in boxes   $\Lambda\subset\Z^d$  with short-range interactions  in the thermodynamic limit $\Lambda\nearrow \Z^d$. We assume that the finite volume Hamiltonians  $H_0^\Lambda(t)$ have a uniformly gapped  spectral island $\sigma_*^\Lambda(t)\subset \sigma(H^\Lambda(t))$ of  uniformly bounded  multiplicity $\kappa^\Lambda(t)$ (typically  the ground state) with corresponding spectral projection $P_*^{\Lambda}(t)$ 
 and that w$^*$-$\lim_{\Lambda\nearrow\Z^d}   \frac{1}{\kappa^{\Lambda}(t)} P_*^{\Lambda}(t) = \rho_0(t)$ exists as a state on the algebra $\mathcal{A}$ of quasi-local observables of the infinite system for all $t\in I$.    The perturbation $V^\Lambda(t)$ can be either a Hamiltonian with short range interactions or a  possibly unbounded external potential, or a sum of both.

Our main result is a generalised super-adiabatic theorem   for the automorphism group $\mathfrak{U}^{\epsi, \eta}_{t,t_0}$
generated by  the densely defined time-dependent derivation $\mathcal{L}_{H(t)}\Ab{\cdot} := \lim_{\Lambda\nearrow \Z^d}  \frac{ 1}{\eta}[H^{\Lambda}(t),\cdot\,] $ on   $\mathcal{A}$. It follows from taking the thermodynamic limit in recent  results on the adiabatic theorem for the finite-volume dynamics \cite{BDF,MT,Teu17}.

Let us briefly explain what we mean by a generalised super-adiabatic theorem.
 For $\epsi=0$, our result is a `standard' super-adiabatic theorem and establishes the existence of super-adiabatic states $\rho_0^\eta(t)$   on $\mathcal{A}$ close to $\rho_0(t)$, i.e.\ 
\[
|\rho_0^\eta(t)(A) - \rho_0(t)(A) | = \Or (\eta)\,,
\]
such that the adiabatic evolution $\mathfrak{U}^{ \eta}_{t,t_0}$ intertwines the super-adiabatic states to all orders in the adiabatic parameter $\eta$, i.e.\
\begin{equation}\label{standardadi}
|\rho_0^\eta(t_0)(\mathfrak{U}_{t,t_0}^{ \eta} \Ab{A}) - \rho_0^\eta (t)(A)|=\Or(\eta^\infty)\,,
\end{equation}
for all $A$ in a dense subspace $\mathcal{D}\subset \mathcal{A}$.
Here and in the following, for better readability, we write the arguments of (densely defined) linear operators on $\mathcal{A}$ inside the brackets~$\Ab{\cdot}$.

For $\epsi\not=0$ the perturbation $\epsi V^\Lambda(t)$ might close the spectral gap and turn the ground state $\rho_0^\Lambda(t)$ of $H^\Lambda_0(t)$ into a `resonance state' $\Pi^{\Lambda,\epsi}(t)$ for $H^\Lambda(t)$ with life-time of order $\Or(\epsi^{-\infty})$,   which we call non-equilibrium almost-stationary state (NEASS). Think, for example, of a linear potential across the sample $\Lambda$ with slope $\epsi\ll 1$ that closes the spectral gap of $H^\Lambda(t)$ once the linear dimension $L$ of $\Lambda$ is sufficiently large. But on each  local term  in $H^\Lambda_0(t)$ such a slowly varying perturbation  acts approximately like a constant shift in energy and thus the NEASS behaves almost like a gapped ground state.
 In particular, as is known for certain types of resonances \cite{AF,EH}, a generalised adiabatic theorem still holds for NEASS. It states that there are super-adiabatic NEASSs $\Pi^{ \epsi,\eta}(t)$ on $\mathcal{A}$ close to $\Pi^{ \epsi }(t)$ such that 
 the adiabatic evolution $\mathfrak{U}^{\epsi, \eta}_{t,t_0}$   approximately intertwines the super-adiabatic NEASSs in the following sense: for any $n>d$ and  for all $A\in\mathcal{D} $
 \begin{equation}\label{IntroThm}
\left|\Pi^{ \epsi,\eta}(t_0)(\mathfrak{U}_{t,t_0}^{\epsi, \eta}\Ab{A}) - \Pi^{ \epsi,\eta} (t)(A)\right|=\Or\left(\eta^{n-d}+\frac{\epsi^n}{\eta^d}   \right)
\end{equation}
uniformly for $t$ in compact sets.
While for $\epsi=0$ this statement reduces to the standard adiabatic theorem \eqref{standardadi},  for $0<\epsi \ll1$ the right hand side of \eqref{IntroThm} is small if only if also $\eta$ is small, but not too small compared to $\epsi$, i.e.\ $  \epsi^{n/d}\ll \eta\ll1$ for some $n\in\N$.
 Physically, this  means that the adiabatic approximation breaks down when the adiabatic switching occurs at times that exceed the lifetime of the NEASS, an effect that has already been observed in adiabatic theory for resonances before, see, e.g., ~\cite{AF,EH}. 
 In addition, we prove a similar result for any initial state $\tilde \rho_0(t_0)$ that is the weak$^*$ limit of a sequence of rank one projections with range contained in the range of $P_*^{\Lambda}(t)$. In the case that 
 there exists different such limits $\tilde \rho_0(t_0)$, an infinite volume version of Kato's parallel transport (including higher order corrections) is required for the construction of super-adiabatic states. 
  
The main motivation for proving such a type of adiabatic theorem stems from linear response theory  and we refer to \cite{HT,MPT} for a   discussion of the problem of rigorously justifying linear response theory for gapped extended systems and the role of our version of the adiabatic theorem in it. In \cite{HT}  we also announced the results of the present paper and put them into context. 

Here we concentrate instead on the mathematical aspects.
Recently, adiabatic theorems for many-body lattice systems have been obtained with error estimates that are uniform in the size of the system, see \cite{BDF,MT,Teu17}. While this uniformity clearly indicates that the corresponding adiabatic theorems should also apply to the infinite system dynamics, i.e.\  after a thermodynamic limit has been taken, on a technical level this is not a simple corollary. 
In short, our main new contribution is to show that in the expansion established in \cite{Teu17} (which generalises \cite{BDF}  from ground states to NEASSs) all relevant objects have a thermodynamic limit. 
 For this purpose we adapt recent results   from \cite{NSY} to our setting.
Moreover, as a second innovation compared to \cite{BDF,MT,Teu17} we   construct super-adiabatic approximations to all orders in~$\epsi$ and $\eta$ by performing suitable resummations of the objects constructed in the finite volume results.

Several of the technical tools developed in the present paper will also be instrumental  for the proof of  an adiabatic theorem for systems  with a gap only in the bulk, which we present in the sequel   \cite{HT2} to this paper. 
  More precisely, in  \cite{HT2}  we first prove  an adiabatic theorem for the evolution of the infinite system  $\mathfrak{U}^{\epsi, \eta}_{t,t_0}$     that holds   under the  condition that the GNS Hamiltonian of  the infinite volume derivation  $\mathcal{L}_{H_0(t)}$  has a spectral gap,  i.e.\ without assuming that a uniformly gapped  sequence of finite volume Hamiltonians $H_0^\Lambda(t)$ with $\lim_{\Lambda\nearrow\Z^d} \mathcal{L}_{H_0^\Lambda(t)}=\mathcal{L}_{H_0(t)}$ exists.
  Then we show that under additional assumptions on the convergence of ground states  an adiabatic theorem holds also for finite systems with a gap in the bulk.
 This result then covers, for example, also the situation of quantum Hall systems with edges, where edge states close the spectral gap but a gap in the bulk remains.  
 
We conclude this introduction with a brief plan of the paper.
In Section~\ref{setting} we describe the mathematical framework and introduce the relevant objects and spaces. In Section~\ref{spacetimeadiabatic} we first formulate the main results as announced in \cite{HT} and briefly discuss them. Then we state the most general adiabatic theorem of this work, which    is proved in Section~\ref{proofs}.  
In order to keep the proofs manageable and transparent, we moved many technical arguments and some preliminaries to several appendices: in Appendix~\ref{finitevolumeresults} we state the main result of \cite{Teu17} in the form we use it in our proof in Section~\ref{proofs}; in Appendix~~\ref{infinitevolumedynamics} we adapt results of \cite{NSY} on the existence of the infinite volume dynamics to our needs in order to define spaces of operators that `have a   thermodynamic limit'; Appendix~\ref{conditionalexpectation} briefly recalls a few facts about conditional expectations for fermionic operators; in Appendix~\ref{technicalApp} we prove the main new technical tools needed for our result, namely that the operations used in the construction of the super-adiabatic NEASS leave spaces of operators that have a  thermodynamic limit invariant;
in Appendix~\ref{resummation} we prove the technical results needed for implementing the resummation procedure; and finally in Appendix~\ref{localtoflocalizedappendix} we prove some facts about certain subspaces of the quasi-local algebra of the infinite system that allow us to extend the main results to  observables that are not strictly local.

\section{Mathematical framework} \label{setting}
\subsection{The lattice and    algebras of local observables}
 We consider fermions with $r$ spin or other internal degrees of freedom on the lattice $\mathbb{Z}^d$. Let $\mathcal{P}_0(\Gamma) := \set{X \subset \Gamma : \vert X \vert < \infty}$ denote the set of finite subsets of $\Gamma$, where $\vert X \vert $ is the number of elements in $X$, and let 
 \begin{equation*}
\Lambda_k := \set{-k, ..., +k}^d\,.
 \end{equation*}
be the centred box of size $2k$. With each $\Lambdak$ we associate a metric $d^{\Lambdak}(\cdot,\cdot)$   which may differ from the standard $\ell^1$-distance $d(\cdot, \cdot )$ on $\Gamma$ restricted to $\Lambdak$, e.g.\ if one considers discrete tube or torus geometries, but satisfies the bulk-compatibility condition
\begin{equation*}  
\forall k \in \mathbb{N} \ \forall x,y \in \Lambdak: d^\Lambdak(x,y) \le d(x,y) \ \  \text{and} \ \   d^{\Lambdak}(x,y) = d(x,y) \ \text{whenever} \ d(x,y) \le k . 
\end{equation*}  
E.g., for a torus geometry, opposite points on the boundary of $\Lambdak$ are considered neighbours and their distance is set to one, while for a cube geometry their distance is $2k$. 

Now, for each $X \in \mathcal{P}_0(\Gamma)$, the corresponding one-particle Hilbert space is $\mathfrak{h}_{X} := \ell^2(X,\mathbb{C}^r)$, the $N$-particle Hilbert space is its $N$-fold anti-symmetric tensor product $\mathfrak{H}_{X,N} := \bigwedge_{j=1}^N \mathfrak{h}_{X}$ and the fermionic Fock space is $\mathfrak{F}_X := \bigoplus_{N=0}^{{r} \vert X\vert} \mathfrak{H}_{X,N}$, where $\mathfrak{H}_{X,0} := \mathbb{C}$. All these Hilbert spaces are finite dimensional and thus all linear operators on them are bounded. The local $C^*$-algebras $\mathcal{A}_{X} := \mathcal{L}(\mathfrak{F}_X)$ are generated by the identity element $\mathbf{1}_{\mathcal{A}_X}$ and the creation and annihilation operators $a_{x,i}^*, a_{x,i}$ for $x \in X$ and $1 \le i \le r$, which satisfy the canonical anti-commutation relations (CAR), i.e.
\begin{equation*}
\{a_{x,i},a_{y,j}\} = \{a_{x,i}^*,a_{y,j}^*\}=0, \;\{a_{x,i},a_{y,j}^*\} = \delta_{x,y}\delta_{i,j}\mathbf{1}_{\mathcal{A}_X}, \;   \forall x,y \in X, \ 1\le i,j\le   r. 
\end{equation*}
Here, $\{A,B\} := AB+BA$ denotes the anti-commutator of $A$ and $B$. If we have $X\subset X'$, $\mathcal{A}_{X}$ is naturally embedded as a subalgebra of $\mathcal{A}_{X'}$. For any $X \in \mathcal{P}_0(\Gamma)$, one defines the parity automorphism
\begin{equation*}
\sigma_X(A) := (-1)^{N_X} A (-1)^{N_X}, \quad A \in \mathcal{A}_X,
\end{equation*}
where 
\begin{equation*}  
N_X := \sum_{x \in X} a^*_x \hspace{-0.5mm}\cdot \hspace{-0.5mm} a_x := \sum_{x \in X} \sum_{i=1}^{r}a^*_{x,i} a_{x,i}. 
\end{equation*}
is the number operator, which has integer spectrum.

For the infinite system, the algebra of local observables is defined as the inductive limit 
\begin{equation*}
\mathcal{A}_{\mathrm{loc}}  := \bigcup_{X \in \mathcal{P}_0(\Gamma)} 	\mathcal{A}_{X}\,,
\qquad\mbox{and its closure}\qquad
\mathcal{A} := 	\overline{\mathcal{A}_{\mathrm{loc}}}^{\Vert \cdot \Vert}
\end{equation*}
with respect to the induced norm $\Vert \cdot \Vert$ is a $C^*$-algebra, called the quasi-local algebra. Due to this structure, there exists a unique automorphism $\sigma: \mathcal{A} \to \mathcal{A}$, for which $\sigma \vert_{\mathcal{A}_X} = \sigma_X$ for any $X \in \mathcal{P}_0(\Gamma)$ and $\sigma^2 = \mathbf{1}_{\mathcal{A}}$. Now, each $A \in \mathcal{A}$ can be uniquely decomposed with respect to the automorphism by
\[
A = A^+ + A^- , \qquad A^{\pm} = \frac{A \pm \sigma(A)}{2}\,.
\]
It follows that $\sigma(A^{\pm })= \pm A^{\pm}$, the even elements $\mathcal{A}^+ \subset \mathcal{A}$ form a $C^*$-subalgebra, but the odd elements $\mathcal{A}^- \subset \mathcal{A}$ do not form a $C^*$-subalgebra. This decomposition is analogously defined for all $X\in \mathcal{P}_0(\Gamma)$. Furthermore, we have $\sigma(\mathcal{A}_{X}) = \mathcal{A}_{X}$ and $[\mathcal{A}_{X}^+,\mathcal{A}_{X'}^+] = \{0\}, [\mathcal{A}_{X}^+,\mathcal{A}_{X'}^-]= \{0\}, \{\mathcal{A}_{X}^-,\mathcal{A}_{X'}^-\} = \{0\}$ whenever $X \cap X' = \emptyset$. The distinction between even and odd elements of the   algebra is in particular relevant for the notion of conditional expectations (see Appendix~\ref{conditionalexpectation}). Also, note that for any $X \in \mathcal{P}_0(\Gamma)$ the set of elements $\mathcal{A}_{X}^N$ commuting with the number operator $N_{X}$, forms a subalgebra of the even subalgebra, i.e. $\mathcal{A}_X^N \subset \mathcal{A}_X^+ \subset \mathcal{A}_X$. 

Consider the following subset of $\mathcal{A}^+$ (cf.\ \cite{MO} for an analogous definition for quantum spin systems), which will be important in the latter. Let $\mathbb{E}_{\Lambdak}: \mathcal{A}^+\to \mathcal{A}_{\Lambdak}^+$ be the conditional expectation on even observables (see Appendix \ref{conditionalexpectation}) and $f: [0,\infty) \to (0, \infty)$ be a bounded, non-increasing function with $\lim\limits_{r \to \infty}f(r) =0 $. For each $A \in \mathcal{A}^+$, let
\begin{equation*}
\Vert A \Vert_f := \Vert A \Vert + \sup_{k \in \mathbb{N}} \left(\frac{\Vert A - \mathbb{E}_{\Lambdak}(A)\Vert }{f(k)}\right). 
\end{equation*}
We denote by $\mathcal{D}_f^+$ the set of all $A \in \mathcal{A}^+$ such that $\Vert A \Vert_f < \infty$.
  As shown in Lemma~B.1 of \cite{MO}, $\mathcal{D}_f$ is a $*$-algebra and    $(\mathcal{D}_f,\|\cdot\|_f)$  a Banach space.
Although  $(\mathcal{D}_f,\|\cdot\|_f)$ is not a normed algebra, it follows from the proof of Lemma~B.1 in \cite{MO} that multiplication $\mathcal{D}_f\times \mathcal{D}_f\to \mathcal{D}_f$, $(A,B)\mapsto AB$, is continuous. 
 
As only  even observables will be relevant to our considerations, we will drop the superscript `$^+$' from now on and redefine $\mathcal{A} := \mathcal{A}^+$.

\subsection{Interactions and associated operator families}

An {\em interaction on a domain $\Lambda_k$} is a map 
\[
\Phi^{\Lambdak} : \left\{X \subset \Lambdak\right\} \to \mathcal{A}_{\Lambdak}^N\,,\quad  X \mapsto \Phi^{\Lambdak}(X)\in \mathcal{A}_{X}^N\,,
\]
with values in the self-adjoint operators. 
It defines an associated self-adjoint operator   $A^{\Lambdak}\in\mathcal{A}_\Lambdak$ by 
\[
A^{\Lambdak} := A^{\Lambdak}(\Phi) := \sum_{X \subset \Lambdak} \Phi^{\Lambdak} (X) \in \mathcal{A}_{\Lambdak}^N\,.
\]
While any self-adjoint operator $A^{\Lambdak} \in \mathcal{A}_\Lambdak$  is trivially given by an interaction, e.g.\ 
by $ \Phi_A^{\Lambdak}(\Lambdak) = A^\Lambdak$ and $\Phi_A^{\Lambdak}(X) = 0$ for $X\not= \Lambdak$, 
the idea is that $\Phi_A^{\Lambdak}(X)$ should be negligible, in a sense to be made precise below, for sets $X$ with large diameter.
 
In order to describe fermionic system on the lattice $\Gamma$ in the thermodynamic limit, we consider  sequences 
$\Phi = \left(\Phi^{\Lambdak}\right)_{k \in \mathbb{N}}$ of interactions on domains $\Lambdak$ 
and call  the whole sequence an {\em  interaction}.  
Finally, by 
an  {\em infinite volume interaction} we mean a map
\[
 \Psi  : \mathcal{P}_0(\Gamma)  \to \mathcal{A}_\mathrm{loc}^N\,,\quad X \mapsto \Psi (X)\in \mathcal{A}_{X}^N \,,
\]
again with values in the self-adjoint operators. Such an interaction for the infinite system defines  an interaction on each domain $\Lambdak$ by restriction, i.e.\ $\Psi^\Lambdak := \Psi|_{\mathcal{P}_0(\Lambdak) }$, and thus also an interaction $\Psi = \left(\Psi^{\Lambdak}\right)_{k \in \mathbb{N}}$ in the general sense.
While for many questions it suffices to consider exclusively infinite volume interactions, we need to work with the more general concept of interactions for two reasons. First, they  allow us to implement non-trivial boundary conditions. E.g., for periodic boundary conditions a hopping term $a^*_{  (k , x_2,\ldots,x_d)  } \hspace{-0.5mm} \cdot \hspace{-0.5mm}a_{ (-k , x_2,\ldots,x_d)  }$   might only appear in the interaction for the Hamiltonian for that specific value of $k$ in order to connect opposite points on the boundary of $\Lambdak$.    Second, even if we would assume that the physical Hamiltonian is defined in terms of  an infinite volume interaction, the algebraic operations involved in the adiabatic perturbation theory lead to other operators that are no longer defined in terms of infinite volume interactions but in terms of general interactions. 
 
Note that  the maps $\Phi^{\Lambdak} : \left\{X \subset \Lambdak\right\} \to \mathcal{A}_{\Lambdak}^N$ can  be extended to maps on all of  $\mathcal{P}_0(\Gamma)$ by declaring $\Phi^{\Lambdak}(Z) := 0$, whenever $Z \cap (\Gamma \setminus \Lambdak) \neq \emptyset$. This new mapping is called the extension of $\Phi^{\Lambdak}$ and is denoted by the same symbol. Similarly, given $\Phi^{\Lambdak}$ and $\Lambdal \subset \Lambdak$, we define  the restriction of   $\Phi^{\Lambdak}$  to $\Lambdal$ by 
\begin{equation*}
\Phi^{\Lambdak} \vert_{\Lambdal} : \{X \subset \Lambdal\} \to  \mathcal{A}_{\Lambdal}^N\,,\quad 	\Phi^{\Lambdak} \vert_{\Lambdal}(X)  := \Phi^{\Lambdak}(X) \,.
\end{equation*}

In order to   control the range of an interaction, i.e.\ to make precise the  additional locality properties of interactions alluded to before, 
 one introduces the following functions. Let 
\begin{equation*}
 F_{\zeta}(r) := \frac{\zeta(r)}{(1+r)^{d+1}}
\end{equation*}
where $\zeta:[0,\infty) \to (0,\infty)$ is a bounded, non-increasing function, which is logarithmically superadditive, i.e.
\begin{equation*}
\zeta(r+s) \ge \zeta(r) \zeta(s), \quad \forall r,s \in [0,\infty)\,. 
\end{equation*}
In particular, the constant function mapping to $1$ satisfies these properties.
 The functions $F_{\zeta}$ are summable and one defines
\begin{equation*}
\Vert F_{\zeta} \Vert_{\Gamma} := \sup_{k \in \mathbb{N}} \sup_{y\in \Lambdak} \sum_{x \in \Lambdak} F_{\zeta}(d^{\Lambdak}(x,y)) < \infty \,.
\end{equation*}
Moreover, they have a finite convolution constant defined by
\begin{equation*}
C_{\zeta} := \sup_{k \in \mathbb{N}} \sup_{x,y\in \Lambdak} \sum_{z \in \Lambdak} \frac{F_{\zeta}(d^{\Lambdak}(x,z)) \ F_{\zeta}(d^{\Lambdak}(z,y))}{F_{\zeta}(d^{\Lambdak}(x,y))} < \infty\,. 
\end{equation*}
 For any such $\zeta$ and $n\in \mathbb{N}_0$, a norm on the vector space of interactions is  \begin{equation}\label{normdefinition}
	\Vert \Phi \Vert_{\zeta,n} :=  \sup_{k \in \mathbb{N}}\sup_{x,y \in \Gamma} \sum_{\substack{ X \in \mathcal{P}_0(\Gamma): \\x,y\in X}} d^\Lambdak\mbox{-diam}(X)^n \frac{\Vert \Phi^{\Lambdak}(X)\Vert}{F_{\zeta}(d^{\Lambdak}(x,y))}\,.
\end{equation}
Note that these norms depend on the sequence of metrics $d^\Lambdak$ on the cubes $\Lambdak$, i.e.\ on the boundary conditions. While this will in general not be made explicit in the notation, we add a superscript $^\circ$  to the norm and to the normed spaces defined below, if we want to emphasise the use of open boundary conditions, i.e.\ 
$d^\Lambdak \equiv d$. The compatibility condition for the metrics $d^{\Lambdak}$ implies that $\Vert \Psi \Vert_{\zeta, n } \le \Vert \Psi \Vert_{\zeta, n }^\circ $. 

A sequence $A =  (A^\Lambdak)_{k\in\N}$ of operators $A^\Lambdak\in\mathcal{A}_\Lambdak$ 
is now called an \SLT~operator family\footnote{The abbreviation stands for `sum of local terms'. Operators of this kind are often called `local Hamiltonians',   but this could lead to confusion in our considerations, because not all appearing operators with this property play the role of a physical Hamiltonian.}, 
if there exists a corresponding
  interaction $\Phi_A$ with  $\Vert \Phi_A\Vert_{1, 0 } < \infty$.

In order to quantify the difference of interactions ``in the bulk'' we also introduce for any 
interaction $\Phi^\Lambdal$ on the domain $\Lambdal$ and any $\Lambda_M\subset \Lambdal$ the quantity 
  \begin{equation*} 
 \Vert \Phi^\Lambdal \Vert_{\zeta, n,\Lambda_M} := \sup_{x,y \in \Lambda_M} \sum_{\substack{ X \subset \Lambda_M: \\x,y\in X}} \mbox{diam}(X)^n \frac{\Vert \Phi^{\Lambdal}(X)\Vert}{F_{\zeta}(d (x,y))}\,,
\end{equation*}
where $d$ and diam now refer to the   $\ell^1$-distance on $\Gamma$.    

 Let $\mathcal{B}_{\zeta, n }$ be the Banach space of interactions with finite $\Vert \cdot \Vert_{\zeta, n }$-norm. Note that $\mathcal{B}_{\zeta, n } \subset \mathcal{B}_{\zeta, m }$ whenever $n \ge m$ and $\mathcal{B}_{\zeta, 0 } \subset \mathcal{B}_{1,0 }$ for any $\zeta$. 
Now, we consider the class $\mathcal{S}$ of (almost) exponentially fast decaying functions in the sequel: A function $\zeta$ belongs to $\mathcal{S}$ if and only if $\zeta:[0,\infty) \to (0,\infty)$ is bounded, non-increasing, logarithmically superadditive, and decays faster than any inverse polynomial, i.e.
\begin{equation*}
\sup_{r \ge 0} r^n \zeta(r) < \infty, \quad \forall n \in \mathbb{N}_0\,.
\end{equation*}
The prime example for a function $\zeta \in \mathcal{S}$ is the exponential $\zeta(r) = \mathrm{e}^{-ar}$ for some $a >0$. 
Moreover, define
\begin{align*}
\mathcal{B}_{\zeta, \infty } := \bigcap_{n \in \mathbb{N}_0} \mathcal{B}_{\zeta, n }, \quad  
\mathcal{B}_{\mathcal{S},n } := \bigcup_{\zeta \in \mathcal{S}} \mathcal{B}_{\zeta, n }, \quad \mathcal{B}_{\mathcal{S},\infty } := \bigcap_{n \in \mathbb{N}_0} \mathcal{B}_{\mathcal{S}, n }
\end{align*}
as spaces of interactions used in the sequel. The corresponding spaces of operator-families are denoted as $\mathcal{L}_{\zeta, n }$, $\mathcal{L}_{\zeta, \infty }$, $\mathcal{L}_{\mathcal{S},n }$, and $\mathcal{L}_{\mathcal{S},\infty }$. Lemma A.1 in \cite{MT} shows that the spaces $\mathcal{B}_{\mathcal{S},n }$ and therefore also $\mathcal{B}_{\mathcal{S},\infty }$ are indeed vector spaces.

We now consider explicit time-dependence of interactions. Let $I \subset \mathbb{R}$ be an interval. We say that a map $A: I \to \mathcal{L}_{\zeta, n }$ is smooth and bounded whenever it is given by interactions $\Phi_A(t)$ such that the maps 
\begin{equation*}
I \to \mathcal{A}_{X}^N, \quad t \mapsto \Phi_A^{\Lambdak}(t,X)
\end{equation*}
are infinitely differentiable for all $\Lambdak$ and $X \subset \Lambdak$ and for all $i\in\N_0$ there exists $\zeta^{(i)}\in \mathcal{S}$ such that $\zeta^{(0)}=\zeta$ and
\begin{equation*}
\sup_{t\in I} \Vert \Phi_A^{(i)}(t)\Vert_{\zeta^{(i)}, n } < \infty\,.
\end{equation*}
Here $\Phi_A^{(i)}(t) = \left(\frac{\mathrm{d}^i}{\mathrm{d}t^i}\Phi^{\Lambdak}_A(t)\right)_{k \in \mathbb{N}}$ denotes the interaction defined by the term-wise derivatives. The corresponding spaces of smooth and bounded time-depen\-dent interactions and operator families are denoted by $\mathcal{B}_{I, \zeta, n }$ and $\mathcal{L}_{I, \zeta, n }$ and are equipped with the norm $\Vert\Phi \Vert_{I, \zeta, n } := \sup_{t \in I} \Vert \Phi (t)\Vert_{\zeta, n }$. 

We say that $A: I \to \mathcal{L}_{\mathcal{S}, \infty }$ is smooth and bounded if for any $n \in \mathbb{N}_0$ there is a $\zeta_n \in \mathcal{S}$ such that $A: I \to \mathcal{L}_{\zeta_n,n }$ is smooth and bounded and we write $\mathcal{L}_{I, \mathcal{S}, \infty }$ for the corresponding space. 
Note, that an interaction $\Phi \in \mathcal{B}_{\zeta, n }$ can be trivially viewed as an element of $\mathcal{B}_{I,\zeta, n }$ by setting $\Phi(t) = \Phi$ for all $t \in I$. 
\subsection{Lipschitz potentials}
As perturbations we will consider   sequences   of external potentials
$
v = \left( v^{\Lambdak}: \Lambdak \to \mathbb{R}\right)_{ k \in \mathbb{N}} $
that satisfy a Lipschitz condition of the following type, 
	\begin{equation*}
	C_v :=  \sup_{k \in \mathbb{N}} \sup_{\substack{x,y \in \Lambdak \\ x \neq y}} \frac{\vert v^{\Lambdak}(x)- v^{\Lambdak}(y)\vert}{ d^{\Lambdak}(x,y)} < \infty,
	\end{equation*}
and call them for short  {\em Lipschitz potentials}. With a Lipschitz potential $v$ we associate the corresponding operator-sequence $V_v = \left(V_v^{ \Lambdak}\right)_{ k \in \mathbb{N}}$ defined by 
\begin{equation*} 
	V_v^{\Lambdak} := \sum_{x \in \Lambdak} v^{ \Lambdak}(x) \hspace{0.5mm} a^*_x \hspace{-0.5mm}\cdot \hspace{-0.5mm} a_x\,. 
\end{equation*}
The space of Lipschitz potentials is denoted by $\mathcal{V}$.  As above, we add a superscript $^\circ$  to the constant and to the space of Lipschitz potentials, if we want to emphasise the use of open boundary conditions, i.e.\ 
	$d^\Lambda_k \equiv d$. The compatibility condition for the metrics $d^{\Lambda_k}$ implies that $C_v \ge C_v^\circ $. Note that a Lipschitz function $v_\infty:\Gamma\to \R$ defines, again by restriction, a Lipschitz potential in $\mathcal{V}^\circ$ and we write $v_\infty \in \mathcal{V}^\circ$ with a slight abuse of notation. We call $v_\infty$ an \textit{infinite volume Lipschitz potential}. Note that a Lipschitz potential is in general not an \SLT~operator since $\sup_k\sup_x |v^{ \Lambdak}(x)|$ need not be finite. 

We say that a map $v: I \to \mathcal{V}$ is smooth and bounded whenever the maps
\begin{equation*}
I \to \mathbb{R}, \quad t \mapsto v^{\Lambdak} (t,x)
\end{equation*}
are infinitely differentiable for all $k \in \mathbb{N}$ and $x \in \Lambdak$, and we have for all $i \in \mathbb{N}_0$ that 
\begin{equation*}
\sup_{t \in I} C_{v^{(i)}(t)} < \infty\,, 
\end{equation*}
where $v^{(i)}(t) = \left(\frac{\mathrm{d}^i}{\mathrm{d}t^i}v^{\Lambdak}(t)\right)_{k \in \mathbb{N}}$ denotes the Lipschitz potential defined by the term-wise derivatives. 
 The space of  Lipschitz potentials on the interval $I$ such that the map $v: I \to \mathcal{V}$ is  smooth  and bounded in the above sense  is denoted by $\mathcal{V}_I$.   
Note, that a Lipschitz potential $v \in \mathcal{V}$ can be trivially viewed as an element of $\mathcal{V}_I$ by setting $v(t) = v$ for all $t \in I$. 

\subsection{Interactions and potentials with a thermodynamic limit}
As a last ingredient we need the notion of \textit{interactions having a thermodynamic limit} and \textit{Lipschitz potentials having a thermodynamic limit}, which will ensure the convergence of the finite volume dynamics, {see Definition 3.7 in \cite{NSY}. }
\begin{defi}  {\rm (Thermodynamic limit of interactions and potentials)}  \label{cauchydefinition}
	~ 
	\begin{itemize} 
		\item[(a)] A time-dependent interaction $\Phi  \in \mathcal{B}_{I,\zeta, n}$ is said to \textit{have a thermodynamic limit} if there exists an infinite volume interaction $ \Psi \in \mathcal{B}_{I,\zeta, n}^\circ$
		such that 
				\begin{align*} \hspace{-.5cm}
		\forall i \in \mathbb{N}_0, M \in \mathbb{N} : \; \lim_{k\to\infty}\sup_{t \in I}\left\Vert \frac{\D^i}{\D t^i}\left( { \Psi}- \Phi^{ \Lambdak}\right) (t)  \right\Vert_{\zeta^{(i)}, n,\LambdaM} \hspace{-6pt}=0. 
		\end{align*}
	We write $\Phi\stackrel{\rm t.d.}{\rightarrow} \Psi$ in this case.
		A time-dependent interaction $\Phi  \in \mathcal{B}_{I,\mathcal{S}, \infty }$ is said to \textit{have a thermodynamic} limit if for all $n \in \mathbb{N}_0$ exists a $\zeta_n \in \mathcal{S}$ such that $\Phi  \in \mathcal{B}_{I,\zeta_n, n }$ has a thermodynamic limit (similarly for $\mathcal{B}_{I, \zeta, \infty }$ and $\mathcal{B}_{I, \mathcal{S}, n }$). 
		
		An SLT operator family is said to \textit{have a thermodynamic limit} if and only if the corresponding interaction does. 
		\item[(b)] 	A Lipschitz potential 
		$v \in \mathcal{V}_I$
		is said to \textit{have a thermodynamic limit} if there exists an infinite volume Lipschitz potential $v_\infty \in \mathcal{V}_I^\circ$   such that 
		\begin{align*}
		\forall M \in \mathbb{N} \quad \exists K\ge M  \quad \forall k  \ge K,  t\in I: v^{ \Lambdak}(t, \cdot) |_{\LambdaM} = v_{\infty}(t, \cdot) |_{\LambdaM}\,.
		\end{align*}
		Again we write $v\stackrel{\rm t.d.}{\rightarrow} v_\infty$. 
	\end{itemize}
	
\end{defi}

Note that  the interactions in $\mathcal{B}_{I,\zeta, n }$ and the Lipschitz potentials in $\mathcal{V}_I$ that have thermodynamic limit form sub-vectorspaces and that any interaction $\Psi=(\Psi^\Lambdak)$ which arises from restricting an infinite volume interaction has this infinite volume interaction as its thermodynamic limit. We provide an equivalent characterization of interactions having a thermodynamic limit in terms of a Cauchy condition in Lemma \ref{lem:equivalentchar}.

 The following proposition shows that the property of having a thermodynamic limit for the interaction resp.\ potential guarantees also the existence of a thermodynamic limit for the associated evolution operators. Its proof is a simple adjustment of a known argument based on Lieb-Robinson bounds and is given in Appendix \ref{infinitevolumedynamics} (cf.~\cite{BR, BD, NSY}).
\begin{prop}  {\rm (Thermodynamic limit of evolution operators)} \label{tdlofcauchyinteractions} ~\\ 
	Let $H_0 \in \mathcal{L}_{I, \zeta, 0 }$ and $v \in \mathcal{V}_I$ both have a thermodynamic limit, i.e.\ $\Phi_{H_0}\stackrel{\rm t.d.}{\rightarrow} \Psi_{H_0} $ and $v\stackrel{\rm t.d.}{\rightarrow} v_\infty$.
	 Let $H = H_0 + V_v$ 
		and  let $U^{\eta,\Lambdak}(t,t_0)$ denote the evolution family generated by $H^{\Lambdak} (t)$ in scaled time with $\eta >0$, i.e.\ the solution to the Schrödinger equation 
		\begin{equation*} 
		\mathrm{i }\eta\frac{\mathrm{d}}{\mathrm{d}t} U^{\eta,\Lambdak}(t,t_0) = H^{\Lambdak} (t) U^{  \eta,\Lambdak}(t,t_0) 
		\end{equation*}
		with $U^{\eta,\Lambdak}(t_0,t_0) = \mathrm{id}$. Then there exists a  co-cycle of automorphisms $\mathfrak{U}_{t,t_0}^{\eta}:\mathcal{A} \to \mathcal{A}$ such that for all $A \in \Aloc $
		\begin{equation*}
		\mathfrak{U}^{\eta}_{t,t_0}\Ab{A} = \lim\limits_{k \to \infty}  \mathfrak{U}^{\eta,\Lambdak}_{t,t_0}\Ab{A}  := \lim\limits_{k \to \infty} U^{\eta,\Lambdak}(t,t_0)^* A U^{\eta,\Lambdak}(t,t_0)\,.
		\end{equation*}
The co-cycle  $\mathfrak{U}^{\eta}_{t,t_0}$ depends only on $\Psi_{H_0}$ and $v_\infty$
		and  is generated by the time-dependent (closed) derivation  $ (\mathcal{L}_{H(t)}, D(\mathcal{L}_{H(t)}))$ associated with $ H(t)$.
\end{prop}

\subsection{States}

A normalised positive linear functional $\omega$ on the  $C^*$-algebra $\mathcal{A}$ is called a state. Here normalised means $\omega(\mathbf{1}_{\mathcal{A}}) = 1$ and positive means $\omega(A^*A) \ge 0$ for all $A \in \mathcal{A}$.

States $\omega^\Lambdak \in E_{\mathcal{A}_\Lambdak}$ on the algebras $\mathcal{A}_\Lambdak$ of the finite systems are given by positive trace-class operators $\rho^\Lambdak$ with tr$\rho^\Lambdak=1$ via $\mathcal{A}_\Lambdak \ni B \mapsto \mathrm{tr}(\rho^{\Lambdak}B)$.
By employing the conditional expectation $\mathbb{E}_{\Lambdak}$ (see Appendix~\ref{conditionalexpectation}), every such state  $\omega^\Lambdak \in E_{\mathcal{A}_\Lambdak}$ also defines a state $\omega^\Lambdak\circ \mathbb{E}_{\Lambdak} \in E_{\mathcal{A} }$ on $\mathcal{A}$.

 By the Banach-Alaoglu theorem, the set $E_{\mathcal{A}}$ of states  is weak$^*$-compact and  every sequence of states has a convergent subsequence. Hence, the assumptions in the next section concerning  convergence of states could be omitted at the cost of extracting appropriate subsequences.

\section{Main results} \label{spacetimeadiabatic}
\subsection{Assumptions} \label{sec:ass}

 We consider a time-dependent fermionic system on $\Z^d$ described by a time-dependent infinite volume interaction  $\Psi_{H^\epsi} = \Psi_{H_0} + \epsi (V_{v_\infty}  + \Psi_{H_1})$ with $ \Psi_{H_0}, \Psi_{H_1} \in  \mathcal{B}^\circ_{I,\exp(-a \, \cdot),\infty }$ for some $a >0$ and   $v_\infty \in \mathcal{V}_I^\circ$ an infinite volume  Lipschitz potential.
While the statement of our adiabatic theorem  concerns exclusively  the dynamics on the infinite volume, we need to assume the existence of a sequence of finite volume Hamiltonians 
$H_0, H_1 \in \mathcal{L}_{I,\exp(-a \, \cdot),\infty }$ and $v \in \mathcal{V}_I$
with appropriate boundary conditions (encoded in the definition of the norms defining the spaces $ \mathcal{L}$ resp.\ the Lipschitz condition), such that $\Phi_{H_0}\stackrel{\rm t.d.}{\rightarrow} \Psi_{H_0} $, $\Phi_{H_1}\stackrel{\rm t.d.}{\rightarrow} \Psi_{H_1} $ and $v \stackrel{\rm t.d.}{\rightarrow} v_\infty $  and such that the following gap condition holds.    In addition to  the ``standard gap condition'', i.e.\ the requirement  that there exists a spectral island $\sigma_*(t) \subset \sigma(H(t))$ of the time dependent Hamiltonian that has a positive distance to the rest of the spectrum for all times (see e.g.\ \cite{Teu03} Eq.\ (1.5)), we require that the diameter of $\sigma_*(t)$ must be smaller than the gap, that  $\sigma_*(t)$ consists of finitely many eigenvalues of finite multiplicity and  that all  conditions   hold uniformly in the volume.

\begin{quote}
 \textbf{Gap condition for $H_0$.}
There exists   $L \in \mathbb{N}$ such that for all $t \in I$, $k \ge L$ and corresponding $\Lambdak$ the operator $H_0^{\Lambdak}(t)$ has a gapped part $\sigma_*^{\Lambdak}(t) \subset \sigma(H_0^{\Lambdak}(t))$ of its spectrum with the following properties: There exist continuous functions $f_{\pm}^{\Lambdak} : I \to \mathbb{R}$ and constants $g > \tilde{g} >0$ such that 
 \begin{itemize}
 	\item [(i)]$f_{\pm}^{\Lambdak}(t) \in \rho(H_0^{\Lambdak}(t))$ and $f_+^{\Lambdak}(t) - f_-^{\Lambdak}(t) \le \tilde{g}$,
 	\item[(ii)] $[f_-^{\Lambdak}(t), f_+^{\Lambdak}(t)] \cap \sigma(H_0^{\Lambdak}(t)) = \sigma_*^{\Lambdak}(t)$, and 
 	\item[(iii)] $\mathrm{dist}(\sigma_*^{\Lambdak}(t), \sigma(H_0^{\Lambdak}(t)) \setminus  \sigma_*^{\Lambdak}(t)) \ge g$, for all $t \in I$. 
 \end{itemize}
 We denote by $P_*^{\Lambdak}(t)$ the spectral projection of $H_0^{\Lambdak}(t)$ corresponding to the spectrum $\sigma_*^{\Lambdak}(t) $ and assume that there exists $\kappa \in \mathbb{N}$ such that $1 \le \mathrm{dim}\,\mathrm{ran}\,P_*^{\Lambdak}(t) =: \kappa^{\Lambdak}(t) \le \kappa$ for all $t \in I$ and $k \ge L$.
 \end{quote}
 Furthermore, assume that for every $t \in I$ the sequence $\big( \frac{1}{\kappa^{\Lambdak}(t)} P_*^{\Lambdak}(t)\big)_{k \in \mathbb{N}}$  of states converges in the weak$^*$-topology to a state $P_*(t)$ on $\mathcal{A}$, which we call the \textit{gapped limit projection at $t \in I$}.\footnote{In all applications we have in mind, the gapped part of the spectrum is at the bottom of the spectrum and the gapped limit projection is a ground state for $\mathcal{L}_{H_0(t)}$.} 

 \begin{rmk} \label{gappedremark}
If a sequence $\big(\rho_0^{\Lambdak}(t)\big)_{k \in \mathbb{N}}$ of states satisfying $\rho_0^{\Lambdak}(t) = P_*^{\Lambdak}(t)\rho_0^{\Lambdak}(t) P_*^{\Lambdak}(t)$ for all $k \ge L$ has a weak$^*$-limit $\rho_0(t)$, we call it a \textit{gapped limit state at $t \in I$}. In the general adiabatic theorem of Section~\ref{gensec}, Theorem~\ref{existenceofneass2}, we consider also such states as initial states, which requires us to lift the Kato parallel transport also to states on $\mathcal{A}$ in the case that $\rho_0(t) \not = P_*(t)$.
 \end{rmk}

\subsection{Adiabatic theorem  and linear response in the thermodynamic limit}

In this section we formulate a concise version of our adiabatic theorem and its implication for    response theory. In the next section we then formulate the underlying more general but also more technical result. The proofs given in this section just explain how the corresponding result follows from the general result. 

Let us briefly summarise the content of the following adiabatic theorem in a non-technical way.
The standard adiabatic theorem of quantum mechanics states that, at leading order in the adiabatic parameter $\eta$, the   Heisenberg time evolution of the projector onto an initially gapped eigenspace    follows the projector onto time-dependent eigenspace of the time-dependent Hamiltonian as long as the gap remains open. 
 A small in $\eta$ and local-in-time change of the instantaneous eigenspaces, a so-called super-adiabatic transformation, results in a super-adiabatic state, to which the   time evolution adheres to any order in $\eta$. Finally, a perturbation of the instantaneous Hamiltonian, which changes the  eigenspace  or transforms it into a resonance with a long lifetime compared to the adiabatic time scale (a NEASS), can be included. 
 
\begin{thm}  {\rm (Adiabatic theorem for gapped   limit projections)}  \label{existenceofneass3} ~\\
Consider a fermionic system on $\Z^d$ described by a time-dependent infinite volume interaction  $\Psi_{H^\epsi} = \Psi_{H_0} + \epsi (V_{v_\infty}  + \Psi_{H_1})$ satisfying the assumptions spelled out in Section~\ref{sec:ass}.
Now, let $P_*(t)$ be the gapped limit projection at $t \in I$ and $\mathfrak{U}_{t,t_0}^{\varepsilon, \eta}$ be the Heisenberg time-evolution on $\mathcal{A}$ generated by $H^{\varepsilon}(t)$ with adiabatic parameter $\eta \in (0,1]$, cf.   Proposition~\ref{tdlofcauchyinteractions}. 
Then for any $\varepsilon, \eta \in (0,1]$ and $t \in I$ there exists a near-identity automorphism $\beta^{\varepsilon, \eta}(t)$\footnote{By near-identity we mean that $\beta^{\varepsilon, \eta}(t)$ is of the form $\beta^{\varepsilon, \eta}(t) = \mathrm{e}^{\mathrm{i} \varepsilon\mathcal{L}_{S^{\varepsilon, \eta}(t)}}$ for an \SLT-operator $\varepsilon S^{\varepsilon, \eta}(t)$ that has a thermodynamic limit. } such that the super-adiabatic NEASS defined by 
\begin{equation*}
\Pi^{\varepsilon, \eta}(t)(A) := P_*(t)( \beta^{\varepsilon, \eta}(t) \Ab{A})\qquad\mbox{for any } A\in\mathcal{A}
\end{equation*}
has the following properties:
\begin{enumerate}
\item It almost {\bf intertwines the   time evolution}: For any  $n \in\N$, there exists a constant $C_n$ such that for any $ t \in I$ and for all finite $X \subset \Gamma$ and $A \in \mathcal{A}_X \subset \mathcal{A}$
\begin{align} \label{adiabbound}
&\left\vert \Pi^{\varepsilon, \eta}(t_0)(\mathfrak{U}_{t,t_0}^{\varepsilon, \eta}\Ab{A}) - \Pi^{\varepsilon, \eta}(t)(A) \right\vert \nonumber \\ &\qquad \qquad  \le \ C_n \ \frac{\varepsilon^{n+1} + \eta^{n+1}}{\eta^{d+1}}    \left(1+\vert t- t_0\vert^{d+1}\right)   \Vert A \Vert \, \vert X \vert^2.
\end{align}
For any $f \in\mathcal{S}$  the same statement (with a different constant $C_n$) holds also for all  $A \in \mathcal{D}_f$ after replacing $\Vert A \Vert \vert X \vert^2$ by $ \Vert A \Vert_f$. 

\item It is  {\bf local in time}: 
   $\beta^{\varepsilon, \eta}(t)$ depends only on $H^{\varepsilon}$ and its time derivatives at time $t\in I$. 
	\item It is \textbf{stationary} whenever the Hamiltonian is stationary: 
	if for some $t\in I$ all time-derivatives of $H^{\varepsilon} $ vanish at time $t$,
	then $\Pi^{\varepsilon, \eta}({t}) = \Pi ^{\varepsilon, 0}({t})$.
\item It {equals the gapped limit projection of $H_0$} whenever the perturbation vanishes and the Hamiltonian is stationary:  if for some $t\in I$ all time-derivatives of $H^{\varepsilon} $   vanish at time $t$ and $V(t)=0$, 
	  then $\Pi^{\varepsilon, \eta}({t}) = \Pi^{\varepsilon, 0}({t}) = P_*(t)$.
 \end{enumerate}
\end{thm}
\begin{proof}
	The result follows by application of Lemma \ref{resummation3} and Lemma \ref{resummation4} in combination with the triangle inequality and Theorem \ref{existenceofneass2}. The generalisation to $A \in \mathcal{D}_f$ follows from Lemma \ref{localtoflocalized}.
\end{proof}
 The right hand side  in  \eqref{adiabbound}  is asymptotically small if 
\[
1\gg \eta \gg \epsi^\frac{n+1}{d+1} \,.
\]
Thus,   the adiabatic approximation is valid as long as the Hamiltonian changes slowly  on the time-scale defined by the energy gap ($\eta\ll 1$), but fast on the time scale defined by the life-time of the NEASS ($ \eta \gg \epsi^\frac{n+1}{d+1}$ for some $n\in \N$).
 
In order to prepare the application to  linear response theory,  we first determine the asymptotic expansion of the super-adiabatic NEASS~$\Pi^{\varepsilon, \eta}(t)$.
\begin{prop}  {\rm (Asymptotic expansion of $\Pi^{\varepsilon, \eta}(t)$)} \label{asymptoticexp3} ~\\
	Under the assumptions of Theorem \ref{existenceofneass3} , there exist time-dependent closed linear maps $\mathcal{K}_j^{\varepsilon, \eta} : D(\mathcal{K}_j^{\varepsilon, \eta}) \to \mathcal{A}$, $j \in \mathbb{N}_0$, where $\mathcal{A}_{\mathrm{loc}} \subset \mathcal{D}_f\subset D(\mathcal{K}_j^{\varepsilon, \eta})$  for all $f\in\mathcal{S}$ 
	and $\mathcal{K}_j^{\varepsilon,0}$ is independent of $\varepsilon$, such that for any $n\in \mathbb{N}$ there exists a constant $C_{n}$ such that for any finite $X \subset \Gamma$ and $A \in \mathcal{A}_X \subset \mathcal{A}$ it holds that 
	\begin{align} \label{expansionbound}
	\sup_{t \in I} \ \left\vert \Pi^{\varepsilon, \eta}(t)(A)- \sum_{j=0}^{n} \varepsilon^j \, P_*(t)\left(\mathcal{K}_j^{ \varepsilon, \eta}(t)\Ab{A}\right)\right\vert  \le  C_{n} \, (\varepsilon^{n+1} \, +\eta^{n+1})  \, \Vert A \Vert \, \vert X\vert^2.
	\end{align}
	For any $f \in\mathcal{S}$  the same statement (with a different constant $C_n$) holds also for all  $A \in \mathcal{D}_f$ after replacing $\Vert A \Vert \vert X \vert^2$ by $ \Vert A \Vert_f$. 
	\\
	The constant term is $\mathcal{K}_0^{\varepsilon, \eta}(t) = \mathbf{1}_{\mathcal{A}}$, which shows that $\beta^{\varepsilon, \eta}(t)$ is indeed a near-identity automorphism when $\varepsilon, \eta \ll 1$. The linear term $\mathcal{K}_1^{\varepsilon, \eta}(t)$ is a densely defined derivation on $\mathcal{A}$ that satisfies\footnote{See Appendix \ref{technicalApp} for a definition of the \SLT~invers of the Liouvillian $\mathcal{I}$.}
\[ 
	 P_* (t) \left(\mathcal{K}_1^{\varepsilon, \eta}(t)\Ab{A}\right) = \lim\limits_{k\to \infty} \mathrm{i} \ \mathrm{tr}\left(P_*^{\Lambdak}(t) \left[\mathcal{I}_{H_0^{\Lambdak}(t)}\left( \frac{\eta}{\varepsilon} \,\mathcal{I}_{H_0^{\Lambdak}(t)} \left(\dot{H}_0^{\Lambdak}(t)\right)- V^{ \Lambdak} (t) \right), A\right]\right)
	\]
	for every $A \in \mathcal{D}_f$. From this expression it follows that $P_*(t) \left(\mathcal{K}_1^{ \varepsilon, \eta}(t)\Ab{A}\right)$ depends on $V_v$ only through the limiting function $v_{\infty}$. 
\end{prop}
\begin{proof}
	The result follows by application of Lemma \ref{resummation3} in combination with the triangle inequality and Proposition \ref{asymptoticexp2}. 
	The generalisation to $A \in \mathcal{D}_f$ follows from Lemma \ref{localtoflocalized} and Lemma \ref{domainofderivlemma}.
\end{proof}
As mentioned above, an important corollary of our results is a rigorous justification of linear response theory and Kubo's formula in the thermodynamic limit (cf.\ Theorem~4.1 in \cite{Teu17} and Theorem~3.3 in \cite{HT}   for further context). 
\begin{cor} {\rm (Response theory to all orders)}\label{responsetheorytoallorders} ~\\ Consider   time-independent infinite volume interactions  $\Psi_{H_0}$, $\Psi_{H_1}$ and $v_\infty$ that satisfy the assumptions spelled out in Section~\ref{sec:ass}.
 Let $f: \mathbb{R} \to [0,1]$ be a smooth  `switching' function satisfying $f(t) = 0$ for $t \le -1$ and $f(t) = 1$ for $t \ge 0$. Moreover, let $\mathfrak{U}_{t,t_0}^{\varepsilon, \eta,f}$ be the Heisenberg time-evolution on $\mathcal{A}$ generated by $\Psi_{H^\epsi} := \Psi_{H_0} + \epsi f(t) (V_{v_\infty}  + \Psi_{H_1})$ with adiabatic parameter $\eta \in (0,1]$. Let $g \in \mathcal{S}$ and define for $A \in \mathcal{D}_{g}$ the response function 
	\begin{equation*}
	\sigma_A^{\varepsilon, \eta, f}(t ) := P_*(\mathfrak{U}_{t,-1}^{\varepsilon, \eta,f}\Ab{A}) - P_*(A)
	\end{equation*}
	and for $j \in \mathbb{N}$ the $j$th order response coefficient as 
	\begin{equation*}
	\sigma_{A,j}:= P_*(\mathcal{K}_j^{}\Ab{A}), 
	\end{equation*}
	where the $\mathcal{K}_j^{}$'s were defined in Proposition \ref{asymptoticexp3} and 
	\begin{equation*}
	\sigma_{A,1} = - \mathrm{i} \lim\limits_{k\to \infty} \mathrm{tr}\left(P_*^{\Lambdak} \left[\mathcal{I}_{H_0^{\Lambdak}} (V^{ \Lambdak} ), A\right]\right)
	\end{equation*}
	is given by the thermodynamic limit of Kubo's formula. 
	Then for any $n,m \in \mathbb{N}$ and $g \in \mathcal{S}$  there exists a constant $C_{n,m,g}$ independent of $\varepsilon$, such that for any $A \in \mathcal{D}_g \subset \mathcal{A}$ and all $t \ge 0$
	\begin{equation*}
	\sup_{\eta \in [\varepsilon^{m}, \varepsilon^{\frac{1}{m}}]} \left\vert \sigma_A^{\varepsilon, \eta, f}(t )  - \sum_{j=1}^{n} \varepsilon^j \sigma_{A,j}\right\vert \le C_{n,m,g} \, \varepsilon^{n+1}  \left(1+t^{d+1}\right) \,  \Vert A \Vert_g\, . 
	\end{equation*}
\end{cor}
\subsection{General space-time adiabatic theorem}\label{gensec}
We now formulate a  version of the space-time adiabatic theorem that allows for more general initial data: instead of $\rho_0 = P_*(t_0)$ we now merely assume that $\rho_0$ is a gapped limit state at $t_0 \in I$ as defined in Remark~\ref{gappedremark}. As it is well known in adiabatic theory, for  $\rho_0\not= P_*(t_0)$   we also need to control the adiabatic evolution `inside'  $P_*(t)$. For finite systems and at leading order this evolution is Kato's parallel transport within the vector bundle defined by the smooth family of projections $P_*^{\Lambdak}(t)$ with respect to the connection induced by the trivial connection on $\R\times \mathfrak{F}_{\Lambdak}$, the so called Berry connection. Clearly, in order to obtain better than leading order adiabatic approximations, also the parallel transport needs to be adjusted. The following theorem is the thermodynamic limit version of Theorem \ref{existenceofneassfinite} proved in  \cite{Teu17}.
Its proof is given in Section \ref{proofs} and the main results from the previous section follow from it by performing a resummation (see Appendix \ref{resummation}) and a generalisation to observables $A \in \mathcal{D}_f$ (see Appendix \ref{localtoflocalizedappendix}).

\begin{thm}\label{existenceofneass2} {\rm (General space-time adiabatic theorem)}\\ Consider a fermionic system on $\Z^d$ described by a time-dependent infinite volume interaction  $\Psi_{H^\epsi} = \Psi_{H_0} + \epsi (V_{v_\infty}  + \Psi_{H_1})$ satisfying the assumptions spelled out in Section~\ref{sec:ass}. Fix $t_0 \in I$, let $\rho_0$ be a gapped limit state at $t_0 \in I$ (cf.\   Remark~\ref{gappedremark}) and let $\mathfrak{U}_{t,t_0}^{\varepsilon, \eta}$ be the Heisenberg time-evolution on $\mathcal{A}$ generated by $H^{\varepsilon}(t)$ with adiabatic parameter $\eta \in (0,1]$. 
Then for any $\varepsilon, \eta \in (0,1]$ and $t \in I$ there exist two sequences $(\alpha_{n,t,t_0}^{\varepsilon, \eta})_n$ and $(\beta_n^{\varepsilon, \eta}(t))_n$ of families of automorphisms of $\mathcal{A}$, the super-adiabatic parallel transport and the super-adiabatic transformation, such that  
 the states defined by
\begin{equation*}
\Pi_n^{\varepsilon, \eta}(t)(A) :=\rho_0 (  \,  \beta_n^{\varepsilon, \eta}(t) \circ \alpha_{n,t,t_0}^{\varepsilon, \eta} \Ab{A})
\end{equation*}
almost intertwine  the adiabatic time evolution in the following sense: There exists a constant $C_n$ such that for any $ t \in I$ and for all finite $X \subset \Gamma$ and all $A \in \mathcal{A}_X \subset \mathcal{A}$
\begin{equation} \label{eq:adiabthmwithoutresum}
\left\vert \Pi_n^{\varepsilon, \eta}(t_0)(\mathfrak{U}_{t,t_0}^{\varepsilon, \eta}\Ab{A}) - \Pi_n^{\varepsilon, \eta}(t)(A) \right\vert     \le \ C_n \ \frac{\varepsilon^{n+1} + \eta^{n+1}}{\eta^{d+1}} \left(1+\vert t- t_0\vert^{d+1}\right)  \Vert A \Vert \, \vert X \vert^2\,.
\end{equation}
Moreover,
\begin{itemize}
\item[(a)] The super-adiabatic parallel transport  $(\alpha_{n,t,t_0}^{\varepsilon, \eta})_n$ is a co-cycle of automorphisms that satisfies $\alpha_{n,t_0,t_0}^{\varepsilon, \eta} = \mathbf{1}_{\mathcal{A}}$ and $P_*(t_0) \circ\alpha_{n,t,t_0}^{\varepsilon, \eta} = P_*(t)$ for every $n \in \mathbb{N}$ and $t \in I$.  
Thus, when acting with  $ P_*(t_0)$, $\alpha_{n,t,t_0}^{\varepsilon, \eta}$ coincides with the standard parallel transport  (a.k.a.\ spectral flow). 
\item[(b)] The super-adiabatic transformations $\beta_n^{\varepsilon, \eta}(t)$ are near-identity automorphisms   only depending on $H^{\varepsilon}$ and its time derivatives at time $t\in I$. By near-identity we mean that $\beta_n^{\varepsilon, \eta}(t)$ is of the form $\beta_n^{\varepsilon, \eta}(t) = \mathrm{e}^{\mathrm{i} \varepsilon\mathcal{L}_{S^{\varepsilon, \eta}_n(t)}}$ for an \SLT-operator $\varepsilon S^{\varepsilon, \eta}_n(t)$ that has a thermodynamic limit. 

 Furthermore, if
for some $t\in I$ all time-derivatives of $H^{\varepsilon} $ vanish at time $t$, then 
 $\beta_n^{\varepsilon, \eta}(t) = \beta_n^{\varepsilon, 0}(t)$. If, in addition, $V( t)=0$, then  $\beta_n^{\varepsilon, \eta}(t) = \beta_n^{\varepsilon, 0}(t) = \mathbf{1}_{\mathcal{A}}$. 
\end{itemize}
\end{thm}
Note that $\alpha_{0,t,t_0}^{0, \eta}$ is just the thermodynamic limit of Kato's parallel transport within the ground state spaces of $H_0(\cdot)$.

For simplicity, we state the explicit expansion of the super-adiabatic NEASS as the thermodynamic limit version of Proposition \ref{asymptoticexpfinite} only for the case of a gapped limit projection. The general case could be handled as in Theorem~3.3 in \cite{MT}. 
\begin{prop} \label{asymptoticexp2} {\rm (Asymptotic expansion of the super-adiabatic NEASS)}\\
Under the assumptions of Theorem \ref{existenceofneass2} and if $\rho_0 = P_*(t_0)$, there exist time-dependent linear maps $\mathcal{K}_j^{\varepsilon, \eta} : \Aloc  \to \mathcal{A}$, $j \in \mathbb{N}_0$, where $\mathcal{K}_j^{\varepsilon,0}$ is independent of~$\varepsilon$, such that for any $n\ge m \in \mathbb{N}$ there exists a constant $C_{n,m}$ such that for any finite $X \subset \Gamma$ and $A \in \mathcal{A}_X \subset \mathcal{A}$ it holds that 
\begin{align*} 
\sup_{t \in I} \ \left\vert \Pi_n^{\varepsilon, \eta}(t)(A)- \sum_{j=0}^{m} \varepsilon^j \, P_*(t)\left(\mathcal{K}_j^{\varepsilon, \eta}(t)\Ab{A}\right)\right\vert  \le  C_{n,m} \left(\varepsilon^{m+1} \, +\eta^{m+1}\right)   \Vert A \Vert \, \vert X\vert^2.
\end{align*}
The constant term is $\mathcal{K}_0^{\varepsilon, \eta}(t) = \mathbf{1}_{\mathcal{A}}$, which shows that $\beta_n^{\varepsilon, \eta}(t)$ is indeed a near-identity automorphism when $\varepsilon, \eta \ll 1$. The linear term $\mathcal{K}_1^{\varepsilon, \eta}(t)$ is a densely defined derivation on $\mathcal{A}$ that satisfies
\[
P_*(t)\left(\mathcal{K}_1^{\varepsilon, \eta}(t)\Ab{A}\right) = \lim\limits_{k\to \infty} \mathrm{i}\ \mathrm{tr}\left(P_*^{\Lambdak}(t) \left[\mathcal{I}_{H_0^{\Lambdak}(t)}\left( \tfrac{\eta}{\varepsilon} \mathcal{I}_{H_0^{\Lambdak}(t)} \left(\dot{H}_0^{\Lambdak}(t)\right)- V^{ \Lambdak} (t) \right), A\right]\right)
\]
for every $A \in \Aloc $. From this expression follows, that $P_*(t) \left(\mathcal{K}_1^{ \varepsilon,\eta}(t)\Ab{A}\right)$ depends on $V_v$ only through the limiting function $v_{\infty}$. 
\end{prop}

\section{Proofs of the main results} \label{proofs}
In this section  we prove Theorem \ref{existenceofneass2} and Proposition~\ref{asymptoticexp2}. These are the infinite-volume versions of the statements in \cite{Teu17}, which we recall in  Appendix \ref{finitevolumeresults} for convenience of the reader.  We freely use the notation from there. 

We begin with two simple  statements about the convergence of automorphisms, states and linear transformations, which will be used freely in the following.
\begin{lem}\label{automorphismlemm}  {\rm (Thermodynamic limit of compositions of automorphisms)}\\
	Let $(\alpha_k)_k$ and $(\beta_k)_k$ be sequences of automorphisms	\begin{equation*}
		\alpha_k : \mathcal{A}_{\Lambdak} \to \mathcal{A}_{\Lambdak} \quad \text{and} \quad \beta_k : \mathcal{A}_{\Lambdak} \to \mathcal{A}_{\Lambdak}
	\end{equation*}
  that  are strongly convergent on the dense set $\mathcal{A}_{\mathrm{loc}}$, i.e.\  there exist automorphisms $\alpha$ and $\beta$ on $\mathcal{A}$ such that
	\begin{equation*}
		\alpha\Ab{A} = \lim\limits_{k\to \infty} \alpha_k\Ab{A} \quad \text{and} \quad \beta\Ab{A} = \lim\limits_{k \to \infty} \beta_k\Ab{A}
	\end{equation*}
	in norm for all $A \in \mathcal{A}_{\mathrm{loc}}$. Then the sequence of automorphisms $(\alpha_k \circ \beta_k)_k$ is strongly convergent to $\alpha \circ \beta$ on $\mathcal{A}_{\mathrm{\mathrm{loc}}}$. 
\end{lem}
\begin{proof}
Let $l<k$ and $A \in \mathcal{A}_\Lambdal$.  Then we have: 
\begin{align*}
		 \Vert (\alpha_k\, \circ&\, \beta_k) \Ab{A}- (\alpha \circ \beta) \Ab{A} \Vert =  
		 \\ \le \quad  &\Vert \alpha_k\circ ( \beta_k- \beta_l) \Ab{A})\Vert + \Vert (\alpha_k- \alpha) \circ   \beta_l  \Ab{A} \Vert+ \Vert \alpha \circ (\beta_l  -  \beta) \Ab{A}\Vert
		  \\ \le \quad &\Vert  ( \beta_k- \beta_l) \Ab{A})\Vert  +  \Vert (\alpha_k- \alpha) \circ   \beta_l  \Ab{A} \Vert+  \Vert  (\beta_l  -  \beta) \Ab{A}\Vert
  \\ \le \quad &\Vert  ( \beta_k- \beta ) \Ab{A})\Vert  +  \Vert (\alpha_k- \alpha) \circ   \beta_l  \Ab{A} \Vert+  2 \Vert  (\beta_l  -  \beta) \Ab{A}\Vert\,.
		 	\end{align*}
	 Let $\varepsilon > 0 $ and choose $l \in \mathbb{N}$ such that  $2 \Vert  (\beta_l  -  \beta) \Ab{A}\Vert<\epsi$. 
	Then   $\beta_l\Ab{A} \in \mathcal{A}_{\mathrm{loc}}$ and thus 
$		\lim_{k\to\infty} \Vert (\alpha_k\circ \beta_k) \Ab{A}- (\alpha \circ \beta) \Ab{A} \Vert < \varepsilon$.
	Since $\varepsilon>0 $ was arbitrary, the claim follows.
	\end{proof}
\begin{lem} \label{stateautolemm} {\rm (Convergence of states composed with linear transformations)}\\ 
	Let $(\omega_k)_k$ be a sequence of states on $\mathcal{A}$ converging to $\omega$ in the weak$^*$-sense. Furthermore, let $(T_k)_k$ be a sequence of (linear) maps $
		T_k: \mathcal{A}_{\Lambdak} \to \mathcal{A}_{\Lambdak}$ that converges strongly on the dense set $\mathcal{A}_{\mathrm{loc}}$ to a (linear) map $T:\mathcal{A}_{\mathrm{loc}} \to \mathcal{A}$, i.e.\
$\lim_{k\to\infty} \| (T_k - T)\Ab{A}\| = 0$ 
	 for all $A \in \mathcal{A}_{\mathrm{loc}}$. 
	Then the sequence of states $(\omega_k \circ T_k)_k$ converges to $\omega \circ T $ in weak$^*$-sense on~$\mathcal{A}_{\mathrm{loc}}$. 
\end{lem}
\begin{proof} 
	For $A \in \mathcal{A}_\mathrm{loc}$ and $k$ large enough, using the triangle inequality, we find
	\begin{align*}
		\vert (\omega_k \circ T_k)(A) - &(\omega \circ T)(A)\vert 
		\le  \Vert \omega_k \Vert \Vert (T_k - T)\Ab{A}\Vert + \vert (\omega_k - \omega)(T\Ab{A}) \vert \overset{k\to \infty}{\longrightarrow} 0\,. \qedhere
	\end{align*}
\end{proof}

\begin{proof}[Proof of Theorem \ref{existenceofneass2}] 
	We prove the existence of a thermodynamic limit for all automorphisms appearing in Theorem \ref{existenceofneassfinite}. 
	\\[3mm] 
\textbf{(a)} Let $V_n^{\varepsilon, \eta, \Lambdak}(t,t_0)$ denote the solution of the effective parallel transport  equation
\begin{equation} \label{effectiveschrodinger}
\mathrm{i}\eta \frac{\mathrm{d}}{\mathrm{d}t} V_n^{\varepsilon, \eta, \Lambdak}(t,t_0) = H_{\parallel,n}^{\varepsilon, \eta, \Lambdak}(t)\,  V_n^{\varepsilon, \eta, \Lambdak}(t,t_0)	
\end{equation}
with 
\[
H_{\parallel,n}^{\varepsilon, \eta, \Lambdak}(t) := \eta K^{\Lambdak}(t) + \sum_{\mu = 0}^{n} \varepsilon^{\mu}h_{\mu}^{\varepsilon,\eta, \Lambdak}(t)
\]
and initial datum $V_n^{\varepsilon, \eta, \Lambdak}(t_0,t_0) = \mathrm{id}$. We have, in particular, that  
\begin{equation*}
P_n^{\varepsilon, \eta, \Lambdak}(t) = V_n^{\varepsilon, \eta, \Lambdak}(t,t_0) \rho_0^{\Lambdak} V_n^{\varepsilon, \eta, \Lambdak}(t,t_0)^*. 
\end{equation*}
$ H_{\parallel,n}^{\varepsilon, \eta, \Lambdak}(t)$ is obtained from the Hamiltonian $H^{ \epsi, \Lambdak}(t) $ and its time derivatives at time $t$, which is assumed to have a thermodynamic limit, by taking sums of commutators and the \SLT~invers of the Liouvillian (see Theorem \ref{existenceofneassfinite}). So by applying Lemma \ref{cauchy1}, Lemma \ref{cauchy2} and Lemma \ref{cauchy3}, the effective Hamiltonian $ H_{\parallel,n}^{\varepsilon, \eta, \Lambdak}(t)$ has a thermodynamic limit and  the limit 
\begin{equation*}
\alpha_{n,t,t_0}^{\varepsilon, \eta} \Ab{A} := \lim\limits_{k \to \infty} V_n^{\varepsilon, \eta, \Lambdak}(t,t_0)^* \,A\, V_n^{\varepsilon, \eta, \Lambdak}(t,t_0)
\end{equation*}
exists in norm for all $A \in \Aloc $ by application of Proposition \ref{tdlofcauchyinteractions} and defines a co-cycle of   automorphisms $\alpha_{n,t,t_0}^{\varepsilon, \eta}$  of $\mathcal{A}$. \\[3mm]
\textbf{(b)} Also the family of operators $\varepsilon S_n^{\varepsilon,  \eta}(t)$ is obtained from the Hamiltonian $H_\epsi$ and its time derivatives at time $t$  by taking sums of commutators and the \SLT~invers of the Liouvillian (see Theorem \ref{existenceofneassfinite}). So by applying Lemma \ref{cauchy1}, Lemma \ref{cauchy2} and Lemma~\ref{cauchy3}, the family of operators $\varepsilon S_n^{\varepsilon,  \eta}(t)$ has a thermodynamic limit and  the limit
\begin{equation*}
\beta_n^{\varepsilon, \eta}(t)\Ab{A} \equiv \mathrm{e}^{-\mathrm{i}\varepsilon \mathcal{L}_{S_n^{\varepsilon, \eta}(t)}}\Ab{A} = \lim\limits_{k \to \infty} \mathrm{e}^{-\mathrm{i}\varepsilon S^{\varepsilon, \eta, \Lambdak}_n(t)} \,A\, \mathrm{e}^{ \mathrm{i}\varepsilon S^{\varepsilon, \eta, \Lambdak}_n(t)}
\end{equation*}
exists in norm for all $A \in \Aloc $ by application of Proposition \ref{tdlofcauchyinteractions} and defines a family $\beta_n^{\varepsilon, \eta}(t)$ of automorphism  of $\mathcal{A}$.
\\[3mm]
\textbf{(c)} Let $U^{\varepsilon, \eta, \Lambdak}(t,t_0)$ be the solution of the Schrödinger equation 
\begin{equation*}
\mathrm{i } \eta\frac{\mathrm{d}}{\mathrm{d}t} U^{\varepsilon, \eta, \Lambdak}(t,t_0) = H^{\varepsilon, \Lambdak} (t) \,U^{\varepsilon, \eta, \Lambdak}(t,t_0) 
\end{equation*}
with $U^{\varepsilon, \eta, \Lambdak}(t_0,t_0) = \mathrm{id}$ and thus, in particular, 
\begin{equation*}
\rho^{\varepsilon, \eta, \Lambdak}(t) = U^{\varepsilon, \eta, \Lambdak}(t,t_0)  \Pi_{n}^{\varepsilon, \eta, \Lambdak}(t_0)U^{\varepsilon, \eta, \Lambdak}(t,t_0)^*\,,
\end{equation*}
where 
\begin{equation*}
\Pi_n^{\varepsilon, \eta, \Lambdak}(t_0) = \mathrm{e}^{\mathrm{i}\varepsilon S^{\varepsilon, \eta, \Lambdak}(t_0)} \rho_0^{\Lambdak} \mathrm{e}^{-\mathrm{i} \varepsilon S^{\varepsilon, \eta, \Lambdak}(t_0)}.
\end{equation*}
Furthermore, again by   Proposition \ref{tdlofcauchyinteractions}, there exists a unique co-cycle of automorphisms $\mathfrak{U}_{t,t_0}^{\varepsilon, \eta}$ of $\mathcal{A}  $ such that the limit 
\begin{equation*}
\mathfrak{U}_{t,t_0}^{\varepsilon, \eta}\Ab{A} = \lim\limits_{k \to \infty} U^{\varepsilon, \eta, \Lambdak}(t,t_0)^* A U^{\varepsilon, \eta, \Lambdak}(t,t_0)
\end{equation*}
exists in norm for all $A \in \Aloc $. 
\\[3mm]
\textbf{(d)} Now, we use the cyclicity of the trace to shift the time dependence of  $	\rho^{\varepsilon, \eta, \Lambdak}(t)$ and $\Pi_n^{\varepsilon, \eta, \Lambdak}(t)$ to the observable $A \in \mathcal{A}_X$. We get
	\begin{align*}
	 \mathrm{tr}\left(\rho^{\varepsilon, \eta, \Lambdak}(t)A\right)  
	&= \mathrm{tr}\left(U^{\varepsilon, \eta, \Lambdak}(t,t_0)\mathrm{e}^{\mathrm{i} \varepsilon S^{\varepsilon, \eta, \Lambdak}_{n}(t_0)} \rho_0^{\Lambdak} \mathrm{e}^{-\mathrm{i}\varepsilon S^{\varepsilon, \eta, \Lambdak}_{n}(t_0)}U^{\varepsilon, \eta, \Lambdak}(t,t_0)^* A  \right) \\ 
	&=  \mathrm{tr}\left( \rho_0^{\Lambdak} \mathrm{e}^{-\mathrm{i} \varepsilon S^{\varepsilon, \eta, \Lambdak}_{n}(t_0)}U^{\varepsilon, \eta, \Lambdak}(t,t_0)^* A  U^{\varepsilon, \eta, \Lambdak}(t,t_0)\mathrm{e}^{\mathrm{i} \varepsilon S^{\varepsilon, \eta, \Lambdak}_{n}(t_0)}\right) \\
	& \overset{k \to \infty }{\longrightarrow} \quad \rho_{0}\circ\beta_{n}^{\varepsilon, \eta}(t_0)\circ\mathfrak{U}_{t,t_0}^{\varepsilon, \eta}(A)= \Pi_{n}^{\varepsilon, \eta}(t_0)( \mathfrak{U}_{t,t_0}^{\varepsilon, \eta} \Ab{A}), 
	\end{align*}  
	as $\alpha_n^{\varepsilon, \eta}(t_0) = \mathbf{1}_{\mathcal{A}}$, and, analogously, 
	\begin{align*}
	 \mathrm{tr}\left(\Pi^{\varepsilon, \eta, \Lambdak}_n(t) A\right) &= \mathrm{tr}\left(\mathrm{e}^{\mathrm{i} \varepsilon S^{\varepsilon, \eta, \Lambdak}_n(t)} V_n^{\varepsilon, \eta, \Lambdak}(t,t_0) \rho_0^{\Lambdak} V_n^{\varepsilon, \eta, \Lambdak}(t,t_0)^* \mathrm{e}^{-\mathrm{i} \varepsilon S^{\varepsilon, \eta, \Lambdak}_n(t)} A  \right) \\ 
	 &=   \mathrm{tr}\left(  \rho_0^{\Lambdak} V_n^{\varepsilon, \eta, \Lambdak}(t,t_0)^* \mathrm{e}^{-\mathrm{i} \varepsilon S^{\varepsilon, \eta, \Lambdak}_n(t)} A  \mathrm{e}^{\mathrm{i} \varepsilon S^{\varepsilon, \eta, \Lambdak}_n(t)} V_n^{\varepsilon, \eta, \Lambdak}(t,t_0)\right)\\  
	 &   \overset{k \to \infty }{\longrightarrow} \quad \rho_{0}\circ \alpha_{n,t,t_0}^{\varepsilon, \eta}\circ\beta_n^{\varepsilon, \eta}(t) (A)= \Pi_n^{\varepsilon, \eta}(t)(A).  
	\end{align*}
The limits exist by applying Lemma~\ref{automorphismlemm} and Lemma \ref{stateautolemm}. Since the bound in \eqref{adiabboundfinite} is uniform in the system size $k$, we  obtain the result.
The other statements in Theorem~\ref{existenceofneass2} are immediate consequences of the corresponding statements in Theorem \ref{existenceofneassfinite}. 
\end{proof}
In order to show also a the existence of a thermodynamic limit of the asymptotic expansion of the super-adiabatic NEASS, we need to prove that the maps $\mathcal{K}_j^{\varepsilon, \eta, \Lambdak}$ also converge  as $k \to \infty$. Since these maps are sums of nested commutators with SLT operators having a thermodynamic limit (cf.\ Theorem \ref{existenceofneassfinite} and the Lemmata in Appendix~\ref{technicalApp}), we need to show that maps of this type have a thermodynamic limit. 
\begin{lem} \label{cauchy4}  
	For any $n \in \mathbb{N}$, let $A_j \in \mathcal{L}_{I, \zeta, (n-j)d }$, $j= 1, ..., n$ all have a thermodynamic limit. Then, for the (time dependent) maps
	\begin{equation*}
		T_n^{\Lambdak}: \mathcal{A}_{\Lambdak}\to \mathcal{A}_{\Lambdak}:B \mapsto \mathrm{ad}_{A^{\Lambdal}_n}\circ \dots \circ \mathrm{ad}_{A^{\Lambdal}_1}\Ab{B},
	\end{equation*}
	the operators $(T_n^{\Lambdak}\Ab{B})_{k \in \mathbb{N}}$ form a Cauchy sequence (uniformly in time) converging to some $T_n\Ab{B} \in \mathcal{A}$ in norm for all $B \in \Aloc $. The resulting (time dependent) map $T_n: \Aloc  \to \mathcal{A}$ is linear {and, in particular, $T_1$ is a densely defined derivation.}
\end{lem}
\begin{proof}
	Let $l \ge k \ge M$, $X \subset \Lambdak$ and $B \in \mathcal{A}_X$. We omit the dependence on   time in the notation. We have
	\begin{align*}
		 \left\Vert (T_n^{\Lambdal}  \right.- &\left. T_n^{\Lambdak})\Ab{B}\right\Vert   =  \left\Vert \left(\mathrm{ad}_{A^{\Lambdal}_n}\circ \dots \circ \mathrm{ad}_{A^{\Lambdal}_1}  - \mathrm{ad}_{A^{\Lambdak}_n}\circ .\dots \circ \mathrm{ad}_{A^{\Lambdak}_1}\right)\Ab{B} \right\Vert \\ 
		\le \quad &\sum_{j=1}^{n} \left\Vert \mathrm{ad}_{A^{\Lambdal}_n}\circ \dots \circ \mathrm{ad}_{\left(A_j^{\Lambdal}-A_j^{\Lambdak}\right)} \circ \dots \circ \mathrm{ad}_{A^{\Lambdak}_1}\Ab{B}  \right\Vert \\ \le \quad &\sum_{j=1}^{n} 2 \, C_n \Vert B \Vert \vert X \vert^n \Vert F \Vert_{\Gamma}^n    \cdot \prod_{i\neq j}\Vert \Phi_{A_i}\Vert_{I, \zeta, (n-i)d } \cdot { \left\Vert   \Phi_{A_j}^{\Lambdal}-\Phi_{A_j}^{\Lambdak}  \right\Vert_{\zeta, (n-j)d, \LambdaM}}\\  \quad + & \sum_{j=1}^{n} 2 \, C_n \Vert B \Vert \vert X \vert^n \Vert F \Vert_{\Gamma}^{n-1}  \cdot \prod_{i=1}^{n}\Vert \Phi_{A_i}\Vert_{I, \zeta, (n-i)d } \cdot \sup_{x \in X} \sum_{y \in \Lambdal \setminus\LambdaM} F_{\zeta}(d^{\Lambdal}(x,y)). 
	\end{align*}
	The first inequality follows from the triangle inequality, the second one is a combination of Lemma B.3 in \cite{Teu17} and the bound used in Theorem 3.8 (ii) in \cite{NSY}. By choosing $\LambdaM$ large enough, one can 
make the second term arbitrarily small.	
 Since the operators $A_j \in \mathcal{L}_{I,\zeta, (n-j)d }$ have a thermodynamic limit, this holds as well for the first term. So, we have proven that $(T_n^{\Lambdak}\Ab{B})_{k \in \mathbb{N}}$ is indeed a Cauchy sequence in $\mathcal{A}$ having a limit $T_n\Ab{B}$. The linearity of $T_n$ follow easily by considering limits. 
\end{proof}
\begin{proof}[Proof of Proposition~\ref{asymptoticexp2}]
First, the super-adiabatic NEASS has a thermodynamic limit by Theorem \ref{existenceofneass2}. Now, all the maps $\mathcal{K}_j^{\varepsilon, \eta, \Lambdak} : \mathcal{A}_{\Lambdak} \to \mathcal{A}_{\Lambdak}$ are sums of nested commutators with the operators $A^{\varepsilon, \eta, \Lambdak}_1, ... , A^{\varepsilon, \eta, \Lambdak}_j \in \mathcal{L}_{I,\mathcal{S}, \infty}$, which are obtained from the Hamiltonian by taking sums of commutators and inverses of the Liouvillian, so they have a thermodynamic limit. Since the bound in \eqref{expfinite} is uniform in the system size we may write `$\lim_{k \to \infty}$' instead of `$\sup_{\Lambdak:k\ge L}$' and the limits exist by applying Lemma~\ref{cauchy4} and Lemma \ref{stateautolemm}. 
The explicit expression for $\mathcal{K}_1^{\varepsilon, \eta}(t)$ follows from the construction of~$A_1^{\varepsilon, \eta}(t)$.
\end{proof}
 
\appendix
 \section{Time-dependent super-adiabatic state for finite systems} \label{finitevolumeresults}
 In this appendix, we restate the space-time adiabatic theorem and the asymptotic expansion of the super-adiabatic NEASS for finite-lattice systems from~\cite{Teu17}  as required in our proofs. 
 
 \begin{thm}  {\rm (Space-time adiabatic theorem for finite systems)}  \label{existenceofneassfinite} ~\\[1mm]   
 Let $H_0, H_1 \in \mathcal{L}_{I,\exp(-a \, \cdot),\infty }$ and $v \in \mathcal{V}_I$ and assume for $H_0$ the gap condition of Section~\ref{sec:ass}.
  Then there exist two sequences $(h_{\mu}^{\varepsilon, \eta})_{\mu \in \mathbb{N}}$ and  $(A_{\mu}^{\varepsilon, \eta})_{\mu \in \mathbb{N}}$ of \SLT-operator families having the following properties  (cf.~\cite{Teu17} for a more detailed discussion of these objects): 
 	\begin{itemize}
 		\item All $h_{\mu}^{\varepsilon, \eta}$ and $A_{\mu}^{\varepsilon, \eta}$ are polynomials of degree $\mu$ in $\frac{\eta}{\varepsilon}$ with coefficients in $\mathcal{L}_{I,\mathcal{S}, \infty }$. 
		 		\item $[h_{\mu}^{\varepsilon, \eta, \Lambdak}(t), P_*^{\Lambdak}(t)] = 0$ for all $t \in I$, $\mu \in \mathbb{N}_0$ and $\eta \in[0,1]$. 
 		\item If for some $t\in I$ all time-derivatives of  $H^{\varepsilon}$ vanish at time $t$,  then $A_{\mu}^{\varepsilon, \eta}(t) = A_{\mu}^{(0)}(t)$ for all $\mu \in \mathbb{N}$, where $A_{\mu}^{(0)}(t)$ is the constant term of the polynomial. If, in addition,  $V( t)=0$, then $A^{\varepsilon, \eta}_{\mu} (t)= A^{(0)}_{\mu} (t)= 0$. 
 	\end{itemize}
For every $n\in \N$	let
\begin{equation*}
 S_n^{\varepsilon, \eta} := \sum_{\mu = 1}^{n} \varepsilon^{\mu-1} A_{\mu}^{\varepsilon, \eta} \in \mathcal{L}_{I,\mathcal{S}, \infty }
 	\end{equation*}
	and for some $t_0\in I$   consider a family of initial states $\rho_0^{\Lambdak}$ with 
	\[
	\rho_0^{\Lambdak} = P^{\Lambdak}_*(t_0) \rho_0^{\Lambdak} P^{\Lambdak}_*(t_0)\,.
	\]
Let $P_n^{\varepsilon, \eta, \Lambdak}(t)$ be  the solution of the effective parallel transport equation\footnote{Here, $K^\Lambda(t) = \mathcal{I}^\Lambda_{H_0(t),g,\tilde{g}}(\dot{H}_0(t))$ is a self-adjoint \SLT~generator of the parallel transport within the vector-bundle $\Xi_{*,I}$ over $I$ defined by $t \mapsto P_*^\Lambda(t)$. The \SLT~coefficients $h_{\mu}^{\varepsilon, \eta, \Lambda}(t)$ of the effective Hamiltonian generate the dynamics within $\Xi_{*,I}$ and commute with $P_*^\Lambda(t)$. If $\rho_0^\Lambda = P_*^\Lambda(t_0)$, the effective Hamiltonian $\sum_{\mu = 0}^{n} \varepsilon^{\mu}h_{\mu}^{\varepsilon, \eta,\Lambda}(t)$ becomes irrelevant and the unique solution of the equation is $P_*^\Lambda(t)$. }
 	\begin{equation*}
 	\mathrm{i}\eta \frac{\mathrm{d}}{\mathrm{d}t} P_n^{\varepsilon, \eta, \Lambdak}(t) = \left[\eta K^{\Lambdak}(t) + \sum_{\mu = 0}^{n} \varepsilon^{\mu}h_{\mu}^{\varepsilon,\eta, \Lambdak}(t),  P_n^{\varepsilon, \eta, \Lambdak}(t)\right]
 	\end{equation*}
 	with $P_n^{\varepsilon, \eta, \Lambdak}(t_0) = \rho_0^{\Lambdak}$.
Then the super-adiabatic states defined by 
 	\begin{equation*}
 	\Pi_n^{\varepsilon, \eta, \Lambdak}(t):= \mathrm{e}^{\mathrm{i}\varepsilon S_n^{\varepsilon, \eta, \Lambdak}(t)} P_n^{\varepsilon, \eta, \Lambdak}(t) \mathrm{e}^{-\mathrm{i}\varepsilon S_n^{\varepsilon, \eta, \Lambdak}(t)}, 
 	\end{equation*}
 	 almost track  the real time evolution in the following sense: Let $\rho^{\varepsilon, \eta, \Lambdak}(t)$ be the solution of the Schrödinger equation
 	\begin{equation*}
 	\mathrm{i} \eta\frac{\mathrm{d}}{\mathrm{d}t} \rho^{\varepsilon, \eta, \Lambdak}(t) = [H^{\varepsilon,\Lambdak}(t),\rho^{\varepsilon, \eta, \Lambdak}(t)], \quad  \rho^{\varepsilon, \eta, \Lambdak}(t_0) = \Pi_n^{\varepsilon, \eta, \Lambdak}(t_0).  
 	\end{equation*}
 	Then, there exists a constant $C_n$, such that for any   $A \in \mathcal{A}_{X}\subset\mathcal{A}_\mathrm{loc}$, $\eta \in (0,1]$ and all $t \in I$
 \begin{equation}
 \label{adiabboundfinite} 
 	\sup_{ k}  \left\vert \mathrm{tr}\left( (\rho^{\varepsilon, \eta, \Lambdak}(t) -\Pi_n^{\varepsilon, \eta, \Lambdak}(t))A\right)\right\vert   \le  C_n \frac{\varepsilon^{n+1} +  \eta^{n+1}}{\eta^{d+1}}  \left(1+\vert t- t_0 \vert^{d+1}\right)  \Vert A \Vert \vert X \vert^2.
 \end{equation}		
 \end{thm}
 For simplicity, we state the explicit expansion of the super-adiabatic NEASS only for the case of a full projection to the gapped part, i.e. $\rho_0 = P_*(t_0)$, and thus $P_n(t) = P_*(t)$. The general case was discussed  in Theorem~3.3 in \cite{MT}. 
 \begin{prop}  {\rm (Asymp.~exp.~of the super-adiabatic NEASS for finite systems)} \label{asymptoticexpfinite} \\ 
 	Under the assumptions of the Theorem \ref{existenceofneassfinite} and if $\rho_0 = P_*(t_0)$, there exist time-dependent linear maps $\mathcal{K}_j^{\varepsilon, \eta, \Lambdak} : \mathcal{A}_{\Lambdak} \to \mathcal{A}_{\Lambdak}$ for all $j \in \mathbb{N}_0$ such that for all $n \ge m \in \mathbb{N}$ there exists a constant $C_{n,m}$ such that for any local observable $A \in \mathcal{A}_X\subset\mathcal{A}_\mathrm{loc}$ it holds that 
 	\begin{align} \label{expfinite} \nonumber
 	\sup_{t \in I}\sup_{k} &\left\vert \mathrm{tr}\left(\Pi_n^{\varepsilon, \eta, \Lambdak}(t)A\right)- \sum_{j=0}^{m} \varepsilon^j \, \mathrm{tr}\left(P_*^{\Lambdak}(t)\, \mathcal{K}_j^{ \varepsilon, \eta, \Lambdak}(t)\Ab{A}\right)\right\vert \\ &\qquad\qquad\le C_{n,m} (\varepsilon^{m+1} +\eta^{m+1})  \Vert A \Vert \vert X\vert^2.
 	\end{align}
 	Each map $\mathcal{K}_j^{\varepsilon, \eta, \Lambdak} $ is given by a finite sum of nested commutators with the operators $A_{1}^{ \varepsilon, \eta, \Lambdak}(t), ... , A_{j}^{\varepsilon, \eta, \Lambdak}(t)$ constructed above such that $\mathcal{K}_j^{\Lambdak, \varepsilon, 0} $ is independent of $\varepsilon$. Explicitly, the first terms are
 	\begin{align*}
 	\mathcal{K}_0^{\varepsilon, \eta, \Lambdak}(t) = \mathbf{1}_{\mathcal{A}_{\Lambdak}} \qquad\text{and}\qquad \mathcal{K}_1^{ \varepsilon, \eta, \Lambdak}(t)\Ab{\cdot} = \mathrm{-i}[A_{1}^{ \varepsilon, \eta, \Lambdak}(t), \cdot\,].
 	\end{align*}
 \end{prop}
 \begin{rmk}
 	The original finite-lattice adiabatic theorem proved in \cite{Teu17} covers also extensive observables  described by \SLT-interactions, where the trace gets replaced by the trace per unit volume. 
	To cover also this case here, we would need   additional assumptions on the Hamiltonian and the observable in order to guarantee the existence of a thermodynamic limit of a trace per unit volume. 
 \end{rmk}
\section{Infinite volume dynamics} \label{infinitevolumedynamics}
In this appendix, we recall and adjust to our needs the underlying theory for lifting the finite volume dynamics to the thermodynamic limit.
The basis for this are Lieb-Robinson bounds, see \cite{LR}. 
As  
$ 
\mathrm{e}^{\mathrm{i}t\mathcal{L}^{\Lambdak}_{H}} = \mathfrak{U}_{t,0}^{\Lambdak}
$ 
whenever the Hamiltonian is   time-independent, we formulate the following theorems only for $\mathfrak{U}_{t,s}^{\Lambdak}$.   
\begin{thm} {\rm (Lieb-Robinson bounds)}  \label{lrb}  \\
Let $H_0 \in \mathcal{L}_{I, \zeta, 0 }$, $v \in \mathcal{V}_I$ and $H = H_0 + V_v$. Let $k \in \mathbb{N}$, $A \in  \mathcal{A}_{X}$ and $B \in \mathcal{A}_Y$ with $X,Y \subsetneq \Lambdak$ and $X\cap Y = \emptyset$. 
Then we have
\begin{align*}
&\left\Vert \left[\mathfrak{U}_{t,s}^{\Lambdak}\Ab{A},B\right] \right\Vert  \le \  \frac{2 \Vert A \Vert  \Vert B \Vert}{C_{\zeta}}\left(\mathrm{e}^{2C_{\zeta}\vert t-s\vert \Vert\Phi_{H_0} \Vert_{I,\zeta, 0 }}-1\right)   \sum_{ x \in X, y \in Y } F_{\zeta}(d^{\Lambdak}(x,y)). 
\end{align*}
If $H_0 \in \mathcal{L}_{I, \exp(-a \, \cdot), 0 }$ for some $a>0$, we define the Lieb-Robinson velocity via 
\begin{equation} \label{lrvelocity}  
v_a = 2 a^{-1}C_{\exp(-a\cdot)}\Vert\Phi_{H_0} \Vert_{I,a, 0 },
\end{equation}
and get the more transparent bound
\begin{equation*}
\left\Vert \left[\mathfrak{U}_{t,s}^{\Lambdak}\Ab{A},B\right] \right\Vert \le 2\Vert A \Vert  \Vert B \Vert \Vert F_1\Vert_{\Gamma} \ C_{\exp(-a\cdot)}^{-1} \ \min(\vert X \vert, \vert Y \vert) \ \mathrm{e}^{a(v_a\vert t-s\vert - \mathrm{dist}(X,Y))}. 
\end{equation*}
\end{thm}
For the proof see, e.g.,   \cite{BD} or \cite{NSYfermionic}.
\begin{thm} {\rm (Comparison of dynamics on different domains)}  \label{differenceoftwodynamics} ~\\
Let $H_0 \in \mathcal{B}_{I, \zeta, 0 }$ and $v \in \mathcal{V}_I$ both have a thermodynamic limit and let $H = H_0 + V_v$. 
\begin{itemize}
\item[(i)] Fix $M \in \mathbb{N}$. Then there exists $K \ge M$ such that for all $l \ge k \ge K$, $X \subset \LambdaM$, $A \in \mathcal{A}_X$ and $s,t \in I$ we have the bound
\begin{align*}
\left\Vert \left(  \mathfrak{U}_{t,s}^{\Lambdal}\vert_{\LambdaM}  - \mathfrak{U}_{t,s}^{\Lambdak}\vert_{\LambdaM}\right)\Ab{A}\right\Vert   \le & \;2  \, \Vert A \Vert \,  \mathrm{e}^{4|t-s|\Vert\Phi_{H_0}\Vert_{I, \zeta, 0 }} |t-s| \;\times\\ & \ \sup_{t \in I}\left\Vert \left(\Phi_{H_0}^{\Lambdal}-\Phi_{H_0}^{\Lambdak}\right)  (t)\right\Vert_{\zeta,0,\LambdaM}   \ \sum_{\substack{x \in X\\y \in \LambdaM}}  F_{\zeta}(d^{\LambdaM}(x,y))\,.
\end{align*}
Here, we introduced the notation $\mathfrak{U}_{t,s}^{\Lambdak}\vert_{\LambdaM}$ for the dynamics generated by $H^{\Lambdak}\vert_{\LambdaM}$, i.e.  with corresponding interaction $\Phi_{H_0}^{\Lambdak}\big\vert_{\LambdaM}$ and potential $v^{\Lambdak}\big\vert_{\LambdaM}$. 
\item[(ii)] Fix $M \in \mathbb{N}$. For any $k \ge M$,  $X \subset \LambdaM$, $A \in \mathcal{A}_X$ and $s,t \in I$
\[ \hspace{-18pt}
\left\Vert \left(\mathfrak{U}_{t,s}^{\Lambdak} - \mathfrak{U}_{t,s}^{\Lambdak}\vert_{\LambdaM}\right)\Ab{A}\right\Vert \le   2 \, \Vert A \Vert \,  \mathrm{e}^{4|t-s|\Vert\Phi_{H_0}\Vert_{I, \zeta, 0 }} |t-s|     \Vert\Phi_{H_0}\Vert_{I, \zeta, 0} \hspace{-8pt}\sum_{\substack{x \in X\\y \in \Lambdak\setminus \LambdaM}}  \hspace{-6pt} F_{\zeta}(d^{\Lambdak}(x,y))\,.
\]
\end{itemize}
\end{thm}
\begin{proof}
The proof for the first part is an adaption of Theorem 3.4 (i) in \cite{NSY} for our notion of interactions and fermionic systems. The crucial part is that the potential is eventually constant, so the potential part of the two time evolutions coincide. The proof of the second part is known and formulated for fermionic systems in Theorem~5.1~(ii) in \cite{BD}. 
\end{proof}
\begin{proof}[Proof of Proposition \ref{tdlofcauchyinteractions}]
Combining the two estimates of Theorem \ref{differenceoftwodynamics}   shows that $\mathfrak{U}_{t,s}^{\Lambdak}\Ab{A}$ is a Cauchy sequence in $\mathcal{A}$ for all $A \in \Aloc $. Since $\mathcal{A}$ is complete, there exists a limit $\mathfrak{U}_{t,s}\Ab{A}$ and the implicitly defined $\mathfrak{U}_{t,s}$ is a co-cycle of automorphisms (follows by standard arguments, see \cite{SStatistical}) which can be extended to the whole $\mathcal{A}$. The statement on the generator of the automorphism can be obtained analogously to Theorem~3.8~(ii) from \cite{NSY}. Hence, we have proven Proposition \ref{tdlofcauchyinteractions}.
\end{proof}

\section{Conditional expectation for fermionic systems} \label{conditionalexpectation}
In this appendix, we consider the conditional expectation for fermionic systems (cf. \cite{AM}, \cite{NSYfermionic}) and therefore distinguish between odd and even observables. We collect the relevant properties for the convenience of the reader.
\begin{lem} {\rm (Adaption of Lemma 4.1 and Lemma 4.2 from \cite{NSYfermionic})}  ~\\
Let $X \subset Z \in \mathcal{P}_0(\Gamma)$. Then there exists a unit-preserving, completely positive linear map $\mathbb{E}_X^{Z}: \mathcal{A}_{Z} \to \mathcal{A}_{Z}$, called the conditional expectation, with
	\begin{itemize}
		\item[(i)] $\Vert \mathbb{E}_X^{Z} \Vert = 1$;
		\item[(ii)] $\mathbb{E}_X^{Z}(\mathcal{A}_{Z}) = \left\{A+B(-1)^{N_{Z}} : A \in \mathcal{A}_X^+, \, B \in \mathcal{A}_{X}^-\right\} \equiv \mathcal{A}_X^{Z}$;
		\item[(iii)] $\mathbb{E}_X^{Z}\Ab{BAC}=B\,\mathbb{E}_X^{Z}\Ab{A}C$, whenever $A \in \mathcal{A}_{Z}$ and $B,C \in \mathcal{A}_X^{Z}$;
		\item[(iv)] If $X,Y \subset Z$, we have $\mathbb{E}_X^Z \circ \mathbb{E}_Y^Z = \mathbb{E}_Y^Z \circ \mathbb{E}_X^Z = \mathbb{E}^Z_{X\cap Y}$;
		\item[(v)] $\mathbb{E}_X^Z\Ab{A^*} = \mathbb{E}_X^Z\Ab{A}^* $.
	\end{itemize}
\end{lem}
Conditional expectations provide local approximations of quasi-local observables in the case of quantum spin systems. Due to the fermionic structure, this property fails to be satisfied for arbitrary observables (see the discussion in \cite{NSYfermionic}, Chapter 4), but can be established for even observables. 

\begin{lem} {\rm (Adaption of Lemmas 4.2 and  4.3 from \cite{NSYfermionic}, Corollary 4.4 from \cite{NSY})}  \label{partialtraces} \ 
	\begin{itemize}
\item[(a)] The restriction of the conditional expectation to even observables has the following properties: 
\begin{itemize}
	\item[(i)] $\mathbb{E}_X^{Z}(\mathcal{A}^+_{Z}) =  \mathcal{A}^+_{X}$ and $\mathbb{E}_X^{Z}\Ab{A} = A$ for $A \in  \mathcal{A}^+_{X}$; 
	\item[(ii)] Whenever $X \subset Y \subset Z$ and $A \in \mathcal{A}_Y^+$, we have $\mathbb{E}^Z_X\Ab{A} = \mathbb{E}^Y_X\Ab{A}$;
	\item[(iii)] Whenever $A \in \mathcal{A}_Z^+$ satisfies the commutator bound
	\begin{equation*}
	\Vert [A,B] \Vert \le \eta \Vert A \Vert \, \Vert B \Vert \quad \text{for all}\  B \in \mathcal{A}_{Z\setminus X}
	\end{equation*}
	for some $\eta \ge 0$, then
	\begin{equation*}
	\Vert A - \mathbb{E}_X^Z\Ab{A}\Vert \le \eta \Vert A \Vert.
	\end{equation*}
\end{itemize} 
\item[(b)] For any $X \in \mathcal{P}_0(\Gamma)$, there exists a unique bounded map and conditional expectation $\mathbb{E}_X: \mathcal{A}^+ \to \mathcal{A}_X^+$ having the following properties: 
\begin{itemize}
\item[(i)] For all $Z \in \mathcal{P}_0(\Gamma)$ containing $X$, we have that $\mathbb{E}_X\Ab{A} = \mathbb{E}_X^Z\Ab{A}$ for all $A \in \mathcal{A}_Z^+$;
\item[(ii)] Whenever $A \in \mathcal{A}^+$ satisfies the commutator bound
\begin{equation*}
\Vert [A,B] \Vert \le \eta \Vert A \Vert \, \Vert B \Vert \quad \text{for all}\  B \in \bigcup_{\substack{ Y \in \mathcal{P}_0(\Gamma): \\X \cap Y = \emptyset}}\mathcal{A}_{Y}
\end{equation*}
for some $\eta \ge 0$, then
\begin{equation*}
\Vert A - \mathbb{E}_X\Ab{A}\Vert \le \eta \Vert A \Vert.
\end{equation*}
\end{itemize}
	\end{itemize}
\end{lem}
\begin{proof}
The statements under item (a) are proven in \cite{NSYfermionic}. For the statements under item (b): The existence of $\mathbb{E}_X$ follows from the BLT theorem and the compatibility condition for the finite conditional expectations. The commutator bound is obtained by considering an approximating local observable. 
\end{proof}
We conclude, that the conditional expectation on even observables can be used essentially in the same way as for quantum spin systems (cf.\ \cite{NSW, NSY}). 

\section{Operations preserving the thermodynamic limit} \label{technicalApp}
In this appendix, we prove various technical lemmata used in the proofs in Section \ref{proofs}. 
Their content is   that 
\begin{enumerate}
	\item the commutator of two operators having a thermodynamic limit has  a thermodynamic limit  (Lemma \ref{cauchy1}),
	\item the commutator of an operator with a Lipschitz potential both having a thermodynamic limit has a thermodynamic limit   (Lemma \ref{cauchy2}), and
	\item the \SLT~invers $\mathcal{I}_{H_0}$ of the Liouvillian $\mathcal{L}_{H_0}: B \mapsto [H_0,B]$ maps almost exponentially localised operators that have a thermodynamic limit after taking the commutator with $H_0$ to almost exponentially localised operators having a thermodynamic limit (Lemma \ref{cauchy3}).
\end{enumerate}
These operations   are   needed for the construction of the super-adiabatic NEASS. 
  Before turning to the precise formulation of these statements, we give an equivalent characterization of \textit{having a thermodynamic} that will be utilised in their proofs. 
\begin{lem}{\rm (Alternative characterization of having a thermodynamic limit)} \label{lem:equivalentchar} \\
The time-dependent interaction $\Phi \in \mathcal{B}_{I,\zeta,n}$ has a thermodynamic limit if and only if it satisfies the following Cauchy property:
\begin{align*} \hspace{-.0cm}
	\forall i \in \mathbb{N}_0, M \in \mathbb{N}  , \delta>0 \quad \exists  K \ge M  \quad \forall  k,l \ge K: \;\sup_{t \in I}\left\Vert \frac{\D^i}{\D t^i}\left( \Phi^{ \Lambdal}- \Phi^{ \Lambdak}\right) (t)  \right\Vert_{\zeta^{(i)}, n,\LambdaM} \hspace{-6pt} \le \delta. 
	\end{align*} 
\end{lem}
\begin{proof}
Using the triangle inequality, one immediately concludes that any interaction $\Phi \in \mathcal{B}_{I,\zeta,n}$ that has a thermodynamic limit also satisfies the Cauchy property. 

On the other hand, suppose that $\Phi \in \mathcal{B}_{I,\zeta,n}$ satisfies the Cauchy-property and fix $i \in \mathbb{N}_0$ and  $M \in \mathbb{N}$. Then
$
\big(  \frac{\D^i}{\D t^i}\Phi^\Lambdak \big|_{\LambdaM}\big)_{k \in \mathbb{N}}
$
is a Cauchy sequence in the Banach space of time-dependent interactions restricted to $\LambdaM$ equipped with the norm $\sup_{t \in I}\Vert \cdot \Vert_{\zeta^{(i)}, n,\LambdaM}$. 
Hence, there exists a limiting time-dependent interaction
\begin{equation*}
\Psi_i^{ \LambdaM}(t,\cdot) : \set{X \subset \LambdaM} \to \mathcal{A}_{\LambdaM}^N 
\end{equation*}
for any $M \in \mathbb{N}$ which satisfy $  \Psi_i^{ \Lambda_{M'}}(t,\cdot) \big|_{\LambdaM} =\Psi_i^{ \LambdaM}(t,\cdot) $ whenever $M' \ge M$. From this we conclude the existence of an infinite volume interaction
$\Psi_i(t,\cdot): \mathcal{P}_0(\Gamma) \to \Aloc^N$ such that $  \Psi_i(t,\cdot) \big|_{\LambdaM} = \Psi_i^{  \Lambda_{M}}(t,\cdot)$ for any $t \in I$ and thus
\begin{equation} \label{eq:conv}
\lim_{k \to \infty} \sup_{t \in I}\left\Vert \left(\Psi_i- \frac{\D^i}{\D t^i}\Phi^\Lambdak\right)(t) \right\Vert_{\zeta^{(i)},n,\LambdaM}  = 0 \,. 
\end{equation}
Until now, we established the existence of an infinite volume interaction for every time derivative. In order to see that $\Phi \in \mathcal{B}_{I, \zeta, n}$ indeed has a thermodynamic limit, it remains to show that (i) $  \sup_{t \in I}\left\Vert \Psi_i (t) \right\Vert_{\zeta^{(i)},n}^\circ < \infty$  for every $i \in \mathbb{N}_0$ and (ii) $\frac{\D^i}{\D t^i} \Psi_0 (t,X) = \Psi_i(t,X)$ for every $X \in \mathcal{P}_0(\Gamma)$. For (i) we use the bulk-compatibility of the metric and obtain
\begin{align*} 
\left\Vert \Psi_i (t) \right\Vert_{\zeta^{(i)},n}^\circ \; = \; &\sup_{ M \in \mathbb{N}}  \left\Vert \Psi_i (t) \right\Vert_{\zeta^{(i)},n, \LambdaM} \\\;\le\; &\sup_{ M \in \mathbb{N}} \left\Vert \left(\Psi_i- \frac{\D^i}{\D t^i}\Phi^{\Lambda_{2dM}}\right)(t) \right\Vert_{\zeta^{(i)},n,\LambdaM} + \sup_{ M \in \mathbb{N}}  \left\Vert  \frac{\D^i}{\D t^i}\Phi^{\Lambda_{2dM}}(t) \right\Vert_{\zeta^{(i)},n, \LambdaM} \\
\;\le\; &\sup_{ M \in \mathbb{N}}  \left\Vert \left(\Psi_i- \frac{\D^i}{\D t^i}\Phi^{\Lambda_{2dM}}\right)(t) \right\Vert_{\zeta^{(i)},n,\LambdaM} + \left\Vert  \Phi^{(i)}(t) \right\Vert_{\zeta^{(i)}, n } \,.
\end{align*}
The first term is bounded uniformly in $t \in I$ by \eqref{eq:conv} and the second term by assumption. This shows (i) and (ii)   is a simple consequence of the uniformity in $t \in I$ in \eqref{eq:conv}. 
\end{proof}

\begin{lem} \label{cauchy1} {\rm (Thermodynamic limit of commutators of local operators)}\\
	Let $A, B \in \mathcal{L}_{I,\zeta, n+d }$ have a thermodynamic limit. Then the commutator $[A,B] \in \mathcal{L}_{I,\zeta, n }$ also has a thermodynamic limit. In particular, if $A,B \in \mathcal{L}_{I,\mathcal{S}, \infty}$ both have a thermodynamic limit, then $[A,B] \in \mathcal{L}_{I,\mathcal{S}, \infty }$ also has a thermodynamic limit. 
\end{lem}
\begin{proof}
	First, the fact, that $[A,B] \in \mathcal{L}_{I,\zeta, n }$ is a consequence of Lemma 4.5 in \cite{BDF}. Now, it suffices to investigate the norm difference for the commutator without any time-derivative, since the time derivative of a commutator yields
	$\tfrac{\mathrm{d}}{\mathrm{d}t}[A,B] = [\tfrac{\mathrm{d}}{\mathrm{d}t}A,B] + [A,\tfrac{\mathrm{d}}{\mathrm{d}t}B]$.	In the proof, we omit the dependence on  time   and get 
	\begin{align*}
	 \left\Vert  \Phi_{[A^{\Lambdal},B^{\Lambdal}]} \right.&-  \left.\Phi_{[A^{\Lambdak},B^{\Lambdak}]} \right\Vert_{\zeta, n,\LambdaM}\;\le\;  \left\Vert \Phi_{[A^{\Lambdal},B^{\Lambdal}-B^{\Lambdak}]}  \right\Vert_{\zeta, n,\LambdaM} + \left\Vert   \Phi_{[A^{\Lambdal}-A^{\Lambdak},B^{\Lambdak}]} \right\Vert_{\zeta, n,\LambdaM} \\ \le\ & C \  \left(\left\Vert \Phi_A\right\Vert_{I, \zeta, n+d }  \left\Vert \Phi_B^{\Lambdal}-\Phi_B^{\Lambdak} \right\Vert_{\zeta, n+d,\LambdaM}\;+\; \left\Vert \Phi_B\right\Vert_{I, \zeta, n+d } \left\Vert  \Phi_A^{\Lambdal}-\Phi_A^{\Lambdak} \right\Vert_{\zeta, n+d,\LambdaM} \right).
	\end{align*}
	The first inequality follows from triangle inequality, the second one follows from the above mentioned lemma.     Lemma \ref{lem:equivalentchar} then implies the claim. 
\end{proof}
\begin{lem} \label{cauchy2}  {\rm (Thermodynamic limit of commutators with Lipschitz potentials)}\\
	Let $A \in \mathcal{L}_{I,\zeta, n+d+1 }$ and $v \in \mathcal{V}_I$ both have a thermodynamic limit. Then the commutator $[A,V_v] \in \mathcal{L}_{I,\zeta, n }$ also has a thermodynamic limit. In particular, if $A \in \mathcal{L}_{I,\mathcal{S}, \infty }$ has a thermodynamic limit, then $[A,V_v] \in \mathcal{L}_{I,\mathcal{S}, \infty }$ also has a thermodynamic limit. 
\end{lem}
\begin{proof}
	The proof is analogous to the one of Lemma \ref{cauchy1}. Here, we use Lemma~2.1 from~\cite{Teu17}. Again, the dependence on   time   is omitted and we get
	\begin{align*}
	 \left\Vert \Phi_{[A^{\Lambdal},V_v^{\Lambdal}]}  - \Phi_{[A^{\Lambdak},V_v^{\Lambdak}]}  \right\Vert_{\zeta, n,\LambdaM} & \;\le\; \left\Vert  \Phi_{[A^{\Lambdal},V_v^{\Lambdal}-V_v^{\Lambdak}]} \right\Vert_{\zeta, n,\LambdaM} + \left\Vert   \Phi_{[A^{\Lambdal}-A^{\Lambdak},V_v^{\Lambdak}]} \right\Vert_{\zeta, n,\LambdaM} \\ & \;\le\;  \frac{r}{2} \,C_v \left\Vert \Phi_A^{\Lambdal}-\Phi_A^{\Lambdak} \right\Vert_{\zeta, n+d+1,\LambdaM}. 
	\end{align*}
	The first inequality follows again from the triangle inequality, the second one from the proof of the above mentioned Lemma 2.1 in \cite{Teu17} and the fact that $v$ is eventually independent of $k$.   Lemma \ref{lem:equivalentchar} then implies the claim.  
\end{proof}
In contrast to the two lemmata above, proving that the \SLT-invers of the Liouvillian has a thermodynamic limit is considerably more complicated --- the notions of conditional expectations (cf.\ Lemma~\ref{partialtraces}) and local approximations of operators and a combination of two estimates (one ensuring quasi-locality, one ensuring the Cauchy-property from  Lemma \ref{lem:equivalentchar}) are needed. 

Let us briefly recall the definition of the \SLT-invers of the Liouvillian. For any $g >0$ (spectral gap of $H_0$) one can find a real-valued, odd function $W_{g} \in L^1(\mathbb{R})$ satisfying
\begin{equation*}
\sup_{s \in \R } \vert s \vert^n \vert W(s) \vert < \infty, \quad \forall n \in \mathbb{N}_0,
\end{equation*}
and with Fourier transform $\widehat{W_{g}} \in C^{\infty}(\mathbb{R})$ satisfying
\begin{equation*}
\widehat{W_{g}}(\omega) = \frac{-\mathrm{i}}{\sqrt{2\pi}\omega} \quad \text{for} \ \vert \omega \vert \ge g \qquad \text{and} \qquad \widehat{W_{g}}(0) = 0.
\end{equation*}
An explicit function $W_{g}$ having all these properties is constructed in \cite{BMNS}. For our purposes (see the Gap Assumption in Section~\ref{sec:ass}), we need a slightly modified version $W_{g, \tilde{g}}$ of this function, cf.\ \cite{Teu17}: Let $g > \tilde{g}>0$ and $\chi_{g,\tilde{g}}\in C^{\infty}(\mathbb{R})$ a real valued function with $\chi_{g,\tilde{g}}(\omega) =0$ for $\omega \in [-\tilde{g}, \tilde{g}]$ and $\chi_{g,\tilde{g}}(\omega) = 1$ for $\vert \omega \vert \ge g$. Then $W_{g,\tilde{g}}$ defined through its Fourier transform as
\begin{equation*}
\widehat{W}_{g, \tilde{g}} = \chi_{g,\tilde{g}} \cdot \widehat{W}_{g}
\end{equation*}
satisfies, in addition to the properties mentioned above for $W_g$, also $ = 0$ for all $\omega \in [-\tilde{g}, \tilde{g}]$. From now on we will drop the subscripts $g$ and $\tilde{g}$. 

For a local operator $A \in \mathcal{A}_Y$ with $Y \subset \Lambdak$, the \SLT-invers of the Liouvillian is defined as
\begin{equation*}
\mathcal{I}_{H_0^{\Lambdak}}\Ab{A} := \int_{\mathbb{R}} \mathrm{d}s \ W(s) \ \mathrm{e}^{\mathrm{i}s\mathcal{L}^{\Lambdak}_{H_0}}\Ab{A}, 
\end{equation*}
where
\begin{equation*}
\mathrm{e}^{\mathrm{i}s\mathcal{L}^{\Lambdak}_{H_0}}\Ab{A} := \mathrm{e}^{\mathrm{i}sH_0^{\Lambdak}}A\mathrm{e}^{-\mathrm{i}sH_0^{\Lambdak}}
\end{equation*}
is the Heisenberg time-evolution of   $A$ generated by $H_0^{\Lambdak}$ as in Proposition~\ref{tdlofcauchyinteractions}. 
{As in \cite{BMNS} we write a   local decomposition of $\mathcal{I}_{H_0^{\Lambdak}}\Ab{A}$ as}
\begin{equation*}
\mathcal{I}_{H_0^{\Lambdak}}\Ab{A} = \sum_{m=0}^{\infty}\Delta_m^{\Lambdak}\Ab{A}
\end{equation*}
with
\begin{equation*}
\Delta_0^{\Lambdak}\Ab{A} = \int_{\mathbb{R}} \mathrm{d}s \ W(s) \cdot \mathbb{E}^{\Lambdak}_{Y \cap {\Lambdak}}\circ \mathrm{e}^{\mathrm{i}s\mathcal{L}^{\Lambdak}_{H_0}}\Ab{A}\,,
\end{equation*}
where, for $Z \subset Z'$, $\mathbb{E}^{Z'}_Z: \mathcal{A}_{Z'} \to \mathcal{A}_Z $ is the conditional expectation (on even observables!). Furthermore, for $m \ge 1$, set
\begin{equation*}
\Delta_m^{\Lambdak}\Ab{A} := \int_{\mathbb{R}} \mathrm{d}s \ W(s) \left( \mathbb{E}^{\Lambdak}_{Y_m\cap {\Lambdak}} -\mathbb{E}^{\Lambdak}_{Y_{m-1}\cap {\Lambdak}}\right) \circ \mathrm{e}^{\mathrm{i}s\mathcal{L}^{\Lambdak}_{H_0}}\Ab{A}  \,.
\end{equation*}
For the above definition we used, for $\delta \ge 0$, the $\delta$-fattening of a set $Y$,   defined as 
\begin{equation}\label{fattening}
Y_{\delta} = \set{z \in \Gamma : \mathrm{dist}(z,Y) \le \delta}. 
\end{equation}
Note that an interaction for $\mathcal{I}_{H_0}(B)$ is given by
\begin{equation*}
\Phi^{\Lambdak}_{\mathcal{I}_{H_0}(B)}(Z) = \sum_{m=0}^{\infty} \sum_{ Y \subset \Lambdak , Y_m = Z}\Delta_m^{\Lambdak}\Ab{\Phi^{\Lambdak}_{B}(Y)}, 
\end{equation*}
where $\Phi_{B}$ is an interaction for $B$. 
\begin{lem} \label{cauchy3} {\rm (Thermodynamic limit of the inverse Liouvillian)}\\
	  Let $H \in \mathcal{L}_{I,\exp(-a \, \cdot), \infty}$ and $B$ either an SLT-operator in $\mathcal{L}_{I,\mathcal{S}, \infty }$ or a Lipschitz potential. Assume that $H$ and  $[H,B] \in \mathcal{L}_{I,\mathcal{S}, \infty }$ both have a thermodynamic limit.  Then $\mathcal{I}_{H}(B) \in \mathcal{L}_{I,\mathcal{S}, \infty }$ also has a thermodynamic limit.
\end{lem}

\begin{proof}
	First, $\mathcal{I}_{H}(B)(t) $ and all time derivatives are in $\mathcal{L}_{\mathcal{S},\infty }$ uniformly in $t$ by Lemma C.2 in \cite{Teu17} and/or Lemma 4.7 in \cite{BDF} and Proposition 4.1 (iv) in~\cite{BDF}. So we get that $\mathcal{I}_{H}(B) \in \mathcal{L}_{I,\mathcal{S},\infty }$. By the argumentation in the proof of Lemma~C.2 in \cite{Teu17}, we can now restrict our attention to the case of $B \in \mathcal{L}_{I,\mathcal{S}, \infty }$  having a thermodynamic limit. We  use again the characterization from Lemma \ref{lem:equivalentchar} and  carry out the proof for the estimate without any time-derivative (i.e. $i=0$) and argue for general $i \in \mathbb{N}$ afterwards (recall Definition~\ref{cauchydefinition}).
	
Let $n \in \mathbb{N}$ and $B \in \mathcal{L}_{I, \zeta_{}, n+d }$. Moreover, let $\tilde{\zeta} \in \mathcal{S}$ denote the function constructed from $\zeta \in \mathcal{S}$ as in the proof of Lemma~4.7 in \cite{BDF}. With the aid of Lemma A.1 in \cite{MT}, we can assume w.l.o.g.\ that $B \in \mathcal{L}_{I, \tilde{\zeta}_{}, n+d }$. Again, we omit the dependence on   time and find 
	\begin{align*}
	&\left\Vert  \Phi_{\mathcal{I}_{H^{\Lambdal}}(B^{\Lambdal})}- \Phi_{\mathcal{I}_{H^{\Lambdak}}(B^{\Lambdak})} \right\Vert_{\tilde{\zeta},n,\LambdaM} \\ & \qquad \le \;  \left\Vert \Phi_{\mathcal{I}_{H^{\Lambdal}}(B^{\Lambdal} - B^{\Lambdak})}  \right\Vert_{\tilde{\zeta},n,\LambdaM} \ + \ \left\Vert  \Phi_{\mathcal{I}_{H^{\Lambdal}}(B^{\Lambdak})}- \Phi_{\mathcal{I}_{H^{\Lambdak}}(B^{\Lambdak})} \right\Vert_{\tilde{\zeta},n,\Lambda_M}  \;=: \; S_1 + S_2\,. 
	\end{align*}
	Each part is now estimated separately. First, we have
	\begin{equation*}
	S_1 \le C \left\Vert \Phi_B^{\Lambdal}-\Phi_B^{\Lambdak} \right\Vert_{\zeta, n+d,\LambdaM}
	\end{equation*}
	by Lemma 4.7 in \cite{BDF}. For the estimate of $S_2$, recall the definition of the norm in \eqref{normdefinition}. Now, we leave $\sup_{x,y \in \Gamma} \frac{1}{F_{\tilde{\zeta}}(d^{\LambdaM}(x,y))}$ aside and have 
	\begin{align*}
	 \sum_{\substack{Z\subset \LambdaM \\ x,y \in Z}} &d^\LambdaM\text{-diam}(Z)^n \Big\Vert \sum_{m=0}^{\infty} \sum_{\substack{Y \subset \Lambdak \\ Y_m = Z}} \left( \Delta_m^{\Lambdal}- \Delta_m^{\Lambdak}\right)\Ab{\Phi_B^{\Lambdak}(Y)}\Big\Vert \\ \le \quad& \sum_{\substack{Z\subset \LambdaM \\ x,y \in Z}} d^\LambdaM\text{-diam}(Z)^n  \sum_{m=0}^{\infty} \sum_{\substack{Y \subset \Lambdak \\ Y_m = Z}} \left\Vert\left( \Delta_m^{\Lambdal}- \Delta_m^{\Lambdak}\right)\Ab{\Phi_B^{\Lambdak}(Y)}\right\Vert \\= \quad & \sum_{Y \subset \LambdaM} \sum_{m=0}^{\infty} \chi(x,y \in Y_m) \ d^\LambdaM\text{-diam}(Y_m)^n \left\Vert\left( \Delta_m^{\Lambdal}- \Delta_m^{\Lambdak}\right)\Ab{\Phi_B^{\Lambdak}(Y)}\right\Vert. 
	\end{align*}
	As the next step, we estimate the norm in the last line in two ways: one way will ensure the quasi-locality, one will the ensure the Cauchy-property. 
	
	On the one hand, we have by the properties of the conditional expectation (see the proof of Lemma 4.7 (i) in \cite{BDF})
	\begin{equation*}
	\left\Vert\left( \Delta_m^{\Lambdal}- \Delta_m^{\Lambdak}\right)\Ab{\Phi_B^{\Lambdak}(Y)}\right\Vert \le 2 \ \vert Y \vert \ \left\Vert\Phi_B^{\Lambdak}(Y) \right\Vert \  \hat{\zeta}(m) =: \text{I}
	\end{equation*}
	for some $\hat{\zeta}\in \mathcal{S}$.
	On the other hand, we have for any $T>0$ to be chosen later
	\begin{align*}
	&\left\Vert\left( \Delta_m^{\Lambdal}- \Delta_m^{\Lambdak}\right)\Ab{\Phi_B^{\Lambdak}(Y)}\right\Vert \\&\quad\le\; 4  \Vert W \Vert_{\infty} \int_{\vert s \vert \le T} \mathrm{d}s \ \left\Vert\left( \mathrm{e}^{\mathrm{i}s\mathcal{L}^{\Lambdal}_{H_0}} - \mathrm{e}^{\mathrm{i}s\mathcal{L}^{\Lambdak}_{H_0}}\right)\Ab{\Phi_B^{\Lambdak}(Y)}\right\Vert  \;+ \; 8\underbrace{\int_{T}^{+ \infty} \mathrm{d}s\ \vert W(s) \vert}_{=C(T) \overset{T \to \infty}{\longrightarrow}0}\ \left\Vert \Phi_B^{\Lambdak}(Y) \right\Vert 
	\\ &\quad\le\;   8 \,\vert Y \vert  \left\Vert \Phi_B^{\Lambdak}(Y) \right\Vert \ C(T)\ + 
	\\ &\quad\quad \;+ \; 8 \,\vert Y \vert   \left\Vert \Phi_B^{\Lambdak}(Y) \right\Vert\underbrace{\Vert W \Vert_{\infty} \int_{-T}^{+T} \mathrm{d}s \, |s|\mathrm{e}^{4|s|\Vert \Phi_H \Vert_{I,\exp(-a \, \cdot),0 }}}_{=\tilde{C}(T)}    \Vert \Phi_H\Vert_{I,\exp(-a \, \cdot),0}\ \times
	 \\ &\quad \qquad\times \bigg(\sup_{y\in \Lambdal}\sum_{x \in \Lambdal \setminus\LambdaM}F_{\exp(-a \, \cdot)}(d^{\Lambdal}(x,y)) + \sup_{y\in \Lambdak}\sum_{x \in \Lambdak \setminus\LambdaM}F_{\exp(-a \, \cdot)}(d^{\Lambdak}(x,y))\ \bigg) 
	 \\  &\quad\quad + \; 8 \, \vert Y \vert  \left\Vert \Phi_B^{\Lambdak}(Y) \right\Vert \tilde{C}(T) \left\Vert F_{\exp(-a \, \cdot)} \right\Vert \  \left\Vert \Phi_H^{\Lambdal}-\Phi_H^{\Lambdak} \right\Vert_{\exp(-a \, \cdot),0,\LambdaM} \\ &\quad=:\;\text{II} 
	\end{align*}
	The last inequality is a combination of the estimates in Theorem \ref{differenceoftwodynamics} after using triangle inequality. 
	
For the next step, we use   that for 
	 $0\le c  \le \min(a,b)$ and any $p,q>1$ with $\frac{1}{p} + \frac{1}{q} = 1$ we have 
	\begin{equation}\label{resummationrmk}
	c \le a^{\frac{1}{p}} \cdot b^{\frac{1}{q}}\,.
	\end{equation}
For $p=q=2$  we can combine the estimates above in order to obtain
	\begin{align*}
& \hspace{-20pt}\left\Vert\left( \Delta_m^{\Lambdal}- \Delta_m^{\Lambdak}\right)\Ab{\Phi_B^{\Lambdak}(Y)}\right\Vert \le \sqrt{\text{I} \cdot \text{II}} \\ & \le \;  4 \ d^\LambdaM\text{-diam}(Y)^d \ \big\Vert \Phi_B^{\Lambdak}(Y) \big\Vert \ \sqrt{\hat{\zeta}(m)} \ \bigg[C(T)+\tilde{C}(T)   \cdot \bigg( \Vert \Phi_H\Vert_{I,\exp(-a \, \cdot),0 } \,\times \\ \quad \times \ & \bigg(\sup_{y\in \LambdaM}\sum_{x \in \Lambdal \setminus\LambdaMp}F_{\exp(-a \, \cdot)}(d^{\Lambdal}(x,y)) + \sup_{y\in \LambdaM}\sum_{x \in \Lambdak \setminus\LambdaMp}F_{\exp(-a \, \cdot)}(d^{\Lambdak}(x,y)) \ \bigg) \\ \quad  + \ &\Vert F_{\exp(-a \, \cdot)} \Vert \  \left\Vert  \Phi_H^{\Lambdal}-\Phi_H^{\Lambdak} \right\Vert_{\exp(-a \, \cdot),0,\LambdaMp} \bigg)\bigg]^{\frac{1}{2}}. 
\end{align*}
We conclude the proof for $i=0$ by arguing that this estimate yields the required properties. First, $d^\LambdaM\text{-diam}(Y)^d$ raises the decay index from $n$ to $n+d$. Second, the term $\sqrt{\hat{\zeta}(m)}$ (again a function in $\mathcal{S}$!) guarantees the quasi-locality of $\mathcal{I}(B)$ as in the proof of Lemma 4.7 in \cite{BDF}. Third, the term $[...]^{\frac{1}{2}}$ provides the Cauchy-property by choosing $T$ and $\LambdaMp$ appropriate. 
	
	The time-derivatives of  $ \mathcal{I}_{H}(B) $ is 
	\begin{equation*}
\frac{\mathrm{d}}{\mathrm{d}t} \mathcal{I}_{H(t)}(B(t)) \;=\; \mathcal{I}(\dot{B}(t)) + \mathrm{i} \int_{\mathbb{R}} \mathrm{d}s \ s\,W(s) 
\int_0^1 \mathrm{d}u\, \mathrm{e}^{\mathrm{i}u s\mathcal{L}_{H(t)}}\circ \mathcal{L}_{\dot H(t)}\circ
\mathrm{e}^{\mathrm{i}(1-u) s\mathcal{L}_{H(t)}} (B(t))\,.	\end{equation*}
The term $\mathcal{I}(\dot B)$ can be treated as above. 
Also the second term can be estimated by similar arguments:
 We split the integral in two parts with $\vert s \vert \le T$ and $\vert s \vert > T$ for some $T$. In the first part, we deal with the additional integration and the commutator with $\dot{H}(t)$ by employing a local decomposition followed by an application of the Lieb-Robinson bound (Theorem \ref{lrb}) in combination with Lemma \ref{partialtraces}, and Lemma \ref{cauchy1}. In the second part, we use that also $\vert s \vert W(s)$ decays faster than any invers polynomial. Higher derivatives have a similar form and can be estimated analogously with the aid of Theorem \ref{lrb} in combination with Lemma \ref{partialtraces}, and Lemma~\ref{cauchy1}. 
\end{proof}

\section{Resummation of the super-adiabatic NEASS} \label{resummation}
In this appendix, we carry out the resummation of the super-adiabatic NEASS in order to formulate our main results without  dependence on $n$.  
Resumming here means choosing the order $n$ of the approximation as a function of $\epsi$. As we have no  good control on the constant $C_n$ in the bound \eqref{eq:adiabthmwithoutresum}, our resummation provides no improvement in the quality of the bound (e.g.\ as from polynomial to exponential), but its main purpose is to remove the $n$-dependence on the left hand side  of~\eqref{eq:adiabthmwithoutresum} by proving the existence of a single {\it $n$-independent} NEASS (cf.~Lemma \ref{resummation2} below), for which the bound~\eqref{eq:adiabthmwithoutresum} resp.~\eqref{adiabbound} holds for all $n \in \mathbb{N}$.
Such resummations are standard, e.g., in microlocal analysis  \cite{Mar}.

In addition to the general assumptions  of Section~\ref{sec:ass}
 we assume $\rho_0 = P_*(t_0)$ so we don't have to worry about the parallel transport   $\alpha_n$.\footnote{If we would allow a general gapped limit state $\rho_0$ and perform a resummation also for the sequence of automorphisms $\alpha_n$, we would loose the explicit time dependence of our bounds  (see Lemma~C.5 in \cite{MT}),   since the sequence of operators generating $\alpha_n$ is only almost exponentially localised.}

\begin{lem} \label{resummation1} 
	By Theorem \ref{existenceofneassfinite}, we have $\varepsilon S_n^{\varepsilon, \eta} = \sum_{j=1}^{n} \sum_{i=0}^{j}\varepsilon^{i} \eta^{j-i} A_{j,i}$ with $A_{j,i} \in \mathcal{L}_{I,\mathcal{S}, \infty }$. There exists a monotonically decreasing sequence of real numbers $\delta_j \to 0$, such that 
	\begin{equation} \label{resumm_S}
	\varepsilon S^{\varepsilon, \eta} = \sum_{j=1}^{\infty} \chi(\varepsilon / \delta_j) \chi(\eta / \delta_j) \sum_{i=0}^{j}\varepsilon^{i} \eta^{j-i} A_{j,i}
	\end{equation}
	where $\chi = \chi_{[0,1]}$, has the following property: For any $n \in \mathbb{N}$ exists a constant $C_n$ such that $\Vert \varepsilon  S^{\varepsilon} - \varepsilon S^{\varepsilon}_n\Vert_{I,1, 0 } \le C_n \max(\varepsilon, \eta)^n$. In particular, we have $\Vert \varepsilon  S^{\varepsilon} - \varepsilon S^{\varepsilon}_n\Vert_{I,1, 0 } \le C_n (\varepsilon^n + \eta^n)$.
\end{lem}
\begin{proof}
	Since $A_{j,i} \in \mathcal{L}_{I,\mathcal{S}, \infty }$ for all $j \in \mathbb{N}$ and $i = 0,..., j$, we have $A_{j,i} \in \mathcal{L}_{I,1,0 }$ for all $j \in \mathbb{N}$ and $i = 0,..., j$. For simplicity, we use the abbreviation $\Vert \cdot \Vert = \Vert \cdot \Vert_{I,1,0 }$. 
	Choose $\delta_j \to 0$ monotonically decreasing such that
	\begin{equation*}
	\max(\varepsilon, \eta) \chi(\varepsilon/ \delta_j) \chi(\eta/ \delta_j)\max_{i=0,...,j}\Vert A_{j,i} \Vert \le \max_{i=0,...,j}\Vert A_{j,i} \Vert \ \delta_j \le \left(\tfrac{1}{2}\right)^j, 
	\end{equation*}
	e.g.
	\begin{equation*}
	\delta_j \le \min \left\{\frac{1}{\max_{i=0,...,j}\Vert A_{j,i} \Vert} \left(\tfrac{1}{2}\right)^j, \left(\tfrac{1}{2}\right)^j, \delta_1, ... , \delta_{j-1}\right\}. 
	\end{equation*}
	Then 
	\begin{align*}
	& \hspace{-15pt}\Vert \varepsilon S^{\varepsilon, \eta}- \varepsilon S^{\varepsilon, \eta}_n \Vert \\ 
	\le \ &\sum_{j=1}^{n} \sum_{i=0}^{j}\varepsilon^i \eta^{j-i}  (1- \chi(\varepsilon/ \delta_j) \chi(\eta/ \delta_j)) \Vert A_{j,i}\Vert + \sum_{j=n+1}^{\infty} \sum_{i=0}^{j}\varepsilon^i \eta^{j-i} \chi(\varepsilon/ \delta_j) \chi(\eta/ \delta_j)\Vert A_{j,i}\Vert \\ 
	\le \ &\max(\varepsilon, \eta)^n \sum_{j=1}^{n} \underbrace{\max(\varepsilon, \eta)^j}_{\le 1}  \underbrace{\left(\delta_j/\max(\varepsilon,\eta)\right)^n  (1-\chi(\varepsilon/ \delta_j) \chi(\eta/ \delta_j)) }_{\le 1}  \delta_j^{-n} (j+1) \max_{i=0,...,j} \Vert A_{j,i} \Vert\\ &+ \sum_{j= n+1}^{\infty} \max(\varepsilon, \eta)^{j-1}  (j+1)\underbrace{ \max(\varepsilon, \eta ) \max_{i=0,...,j} \Vert A_{j,i} \Vert  \chi(\varepsilon/ \delta_j) \chi(\eta/ \delta_j)}_{\le  \left(\frac{1}{2}\right)^j} \\ \le \ &   \max(\varepsilon, \eta)^n \left(\sum_{j=1}^{n}  \delta_j^{-n} (j+1) \max_{i=0,...,j} \Vert A_{j,i} \Vert+ 2 \right) = C_n \max(\varepsilon, \eta)^n.  \qedhere 
	\end{align*}
\end{proof}
\begin{lem} \label{resummation2}  
	Fix $\varepsilon \in (0,1]$ and let $S^{\varepsilon}$ as above. Then the limits
	\begin{align*}
	&\mathrm{tr}\left(\Pi^{\varepsilon, \eta, \Lambdak}(t) A\right)  =\mathrm{tr}\left(P_*^{\Lambdak}(t) \mathrm{e}^{-\mathrm{i} \varepsilon S^{\varepsilon, \eta, \Lambdak}(t)} A  \mathrm{e}^{\mathrm{i} \varepsilon S^{\varepsilon, \eta, \Lambdak}(t)}\right) \\  &   \overset{k \to \infty }{\longrightarrow} \quad P_*(t)(\beta^{\varepsilon, \eta}(t)\Ab{A}) = \Pi^{\varepsilon, \eta}(t)(A)  
	\end{align*}
	and 
	\begin{align*}
	 \mathrm{tr}\left(\rho^{\varepsilon, \eta, \Lambdak}(t)A\right)  \;  &= \;\mathrm{tr}\left(\Pi^{\varepsilon, \eta, \Lambdak}(t_0)  \, \mathfrak{U}_{t,t_0}^{\varepsilon, \eta, \Lambdak}\Ab{A}\right) \\ \;&=  \;\mathrm{tr}\left( P_*^{\Lambdak}(t_0)\, \mathrm{e}^{-\mathrm{i} \varepsilon S^{\varepsilon, \eta, \Lambdak}(t_0)}\,\mathfrak{U}_{t,t_0}^{\varepsilon, \eta, \Lambdak}\Ab{A} \,\mathrm{e}^{\mathrm{i}\varepsilon S^{\varepsilon, \eta, \Lambdak}(t_0)}\right) \\
	& \overset{k \to \infty }{\longrightarrow} \quad P_*(t_0)(\beta^{\varepsilon, \eta}(t_0)\circ \mathfrak{U}_{t,t_0}^{\varepsilon, \eta, \Lambdak}\Ab{A}) = \Pi^{\varepsilon, \eta}(t_0)( \mathfrak{U}_{t,t_0}^{\varepsilon, \eta, \Lambdak}\Ab{A})
	\end{align*}  
	exist for all $A \in \Aloc $. In particular, the states $\Pi^{\varepsilon, \eta}(t)$ on $\mathcal{A}$ are well defined.
\end{lem}
\begin{proof}
	Since $\delta_j \to 0$, the sum in the definition of $\varepsilon S^{\varepsilon, \eta}$  (see \eqref{resumm_S})  is finite for each $\varepsilon, \eta \in (0,1]$. The existence of the limits follows analogously to the proof of Theorem~\ref{existenceofneass2},  since the SLT-operators generating the dynamics have a thermodynamic limit.
\end{proof}
\begin{lem} \label{resummation3}  
	Let $A \in \mathcal{A}_X$ and $t \in I$. Then, for any $n \in \mathbb{N}$ and $q >0$ there exists a constant $C_{n,q}$ such that 
	\begin{equation*}
	\sup_{t \in I}\left\vert \Pi^{\varepsilon, \eta}(t)(A) - \Pi_n^{\varepsilon, \eta}(t)(A) \right\vert \le C_{n,q} (\varepsilon^{n}+\eta^n) \Vert A \Vert \vert X \vert^{1+\frac{1}{q}}. 
	\end{equation*}
\end{lem}
\begin{proof}
	As in the proof of Theorem 5.2 in \cite{Teu17} we use the cyclicity of the trace, Lemma~C.5 in \cite{MT} in combination with the inequality   \eqref{resummationrmk} and Lemma \ref{resummation1} to obtain
	\begin{align*}
	& \hspace{-20pt}\left\vert \mathrm{tr}\left(\left(\Pi^{\varepsilon, \eta, \Lambdak}(t)  - \Pi_n^{\varepsilon, \eta, \Lambdak}(t)\right) A\right) \right\vert \\ \le  & \  \int_{0}^{1} \mathrm{d}\alpha  \left\vert \mathrm{tr} \left(\mathrm{e}^{ \mathrm{i}\alpha \varepsilon S^{\varepsilon, \Lambdak}(t)} \,P_*^{\Lambdak}(t) \,\mathrm{e}^{-\mathrm{i} \alpha \varepsilon S^{\varepsilon, \eta, \Lambdak}(t)} \right. \right. \\  & \quad \quad \times \left. \left. \left[\left(S^{\varepsilon, \eta, \Lambdak}(t) - S^{\varepsilon, \eta, \Lambdak}_n(t)\right), \mathrm{e}^{-\mathrm{i}(1-\alpha) \varepsilon S_n^{\varepsilon, \eta, \Lambdak}(t)} \,A \,\mathrm{e}^{ \mathrm{i}(1-\alpha) \varepsilon S_n^{\varepsilon, \eta, \Lambdak}(t)}\right]\right) \right\vert \\ \le & \ C_q \,\Vert S^{\varepsilon, \eta, \Lambdak} - S^{\varepsilon, \eta, \Lambdak}_n \Vert_{I,1,0 } \,\Vert A \Vert \, \vert X \vert^{1+\frac{1}{q}} \\ \le & \ \ C_{n,q} \,(\varepsilon^n+ \eta^{n})\, \Vert A \Vert \,  \vert X \vert^{1+\frac{1}{q}}
	\end{align*}
	uniformly in $t \in I$.
	Since the bound is independent of $\Lambdak$ we may take the limit $k \to \infty$, which exists by Lemma \ref{resummation2}, and get the result. 
\end{proof}
\begin{lem} \label{resummation4}  
	Let $A \in \mathcal{A}_X$ and $t,t_0 \in I$. Then, for any $n \in \mathbb{N}$ and $p,q>0$ exists a constant $C_{n,p,q}$ such that 
	\[ 
	 \left\vert \Pi^{\varepsilon, \eta}(t_0)( \mathfrak{U}_{t,t_0}^{\varepsilon, \eta}\Ab{A}) - \Pi_n^{\varepsilon, \eta}(t_0)( \mathfrak{U}_{t,t_0}^{\varepsilon, \eta}\Ab{A})\right\vert \le C_{n,p,q} \, \frac{\varepsilon^{n} + \eta^n}{\eta^{d(1+\frac{1}{q})}}\, (1+\vert t -t_0\vert)^{d(1+\frac{1}{q})} \Vert A \Vert \vert X \vert^{1+\frac{1}{p}+ \frac{1}{q}}\,.
	\]
\end{lem}
\begin{proof}
	For $Z \subset Z'$, let $\mathbb{E}^{Z'}_Z: \mathcal{A}_{Z'} \to \mathcal{A}_Z $ be the conditional expectation (on even observables). 
	Recall the notion of fattening of a set defined in \eqref{fattening} and let $A \in \mathcal{A}_X$ for some $X \subset \Lambdak$. 
	Defining
	\begin{equation*}
	A^{(0)} = \mathbb{E}^{\Lambdak}_{X_{\frac{v}{\eta}\vert t -t_0\vert }\cap \Lambdak}\circ \mathfrak{U}_{t,t_0}^{\varepsilon, \eta, \Lambdak}\Ab{A} , 
	\end{equation*}
	where $v$ is the Lieb-Robinson velocity for $\min(a,a')$ as in \eqref{lrvelocity}, and for $j\ge1$
	\begin{equation*}
	A^{(j)} =\left( \mathbb{E}^{\Lambdak}_{X_{\frac{v}{\eta}\vert t -t_0\vert + j }\cap \Lambdak}  - \mathbb{E}^{\Lambdak}_{X_{\frac{v}{\eta}\vert t -t_0\vert +j-1}\cap \Lambdak}\right)\circ \mathfrak{U}_{t,t_0}^{\varepsilon, \eta, \Lambdak}\Ab{A}
	\end{equation*}
	we can write $\mathfrak{U}_{t,t_0}^{\varepsilon, \eta, \Lambdak}\Ab{A} = \sum_{j=0}^{\infty}A^{(j)}$, where the sum is always finite, since eventually $X_{\frac{v}{\eta}\vert t-t_0 \vert + j }\cap \Lambdak = \Lambdak$. 
	Clearly, $A^{(j)} \in \mathcal{A}_{X_{\frac{v}{\eta}\vert t-t_0 \vert + j }}$ and we bound
	\begin{equation*}
	\vert X_{\frac{v}{\eta}\vert t-t_0 \vert + j } \vert \le \vert X \vert (2(\tfrac{v}{\eta}\vert t -t_0\vert + j))^d \le C \vert X\vert (j+1)^d (1+\eta^{-1}\vert t-t_0\vert)^d. 
	\end{equation*}
	According to the properties of the conditional expectation on even observables (see Lemma \ref{partialtraces}) and using the Lieb-Robinson bound (see Theorem \ref{lrb}) in combination with the inequality  \eqref{resummationrmk}, we have
	\begin{equation} \label{resuminequ}
	\Vert A^{(j)} \Vert \le \min\left(C \Vert A \Vert \vert X\vert \mathrm{e}^{-\min(a,a')j}, \ 2 \Vert A\Vert\right) \le C_p \Vert A \Vert \vert X\vert^{\frac{1}{p}} \mathrm{e}^{-\frac{\min(a,a')}{p}j}
	\end{equation}
	Using the cyclicity of the trace, we are left to estimate
	\begin{align*}
	 \left\vert \mathrm{tr}\left(\left(\Pi^{\varepsilon, \eta, \Lambdak}\right.\right.\right.&\left.\left.\left. \hspace{-9pt}(t_0)  - \Pi^{\varepsilon, \eta, \Lambdak}_n(t_0)\right) \mathfrak{U}_{t,t_0}^{\varepsilon, \eta, \Lambdak}\Ab{A}\right) \right\vert \\ \le \ & \sum_{j=0}^{\infty} \left\vert \mathrm{tr}\left(\left(\Pi^{\varepsilon, \eta, \Lambdak}(t_0)  - \Pi^{\varepsilon, \eta, \Lambdak}_n(t_0)\right) A^{(j)}\right) \right\vert \\ \le \ &C_{n,q} (\varepsilon^{n}+ \eta^n) \sum_{j=0}^{\infty}  \Vert A^{(j)} \Vert \vert X_{\frac{v}{\eta}\vert t -t_0\vert + j } \vert^{1+\frac{1}{q}} \\ \le \ & C_{n,p,q} \, \frac{\varepsilon^{n} + \eta^n}{\eta^{d(1+\frac{1}{q})}} \, (1+\vert t-t_0 \vert)^{d(1+\frac{1}{q})} \sum_{j=0}^{\infty} (j+1)^{d(1+\frac{1}{q})} \mathrm{e}^{- \frac{\min(a,a')}{p} j}  \Vert A \Vert \vert X \vert^{1+\frac{1}{p}+\frac{1}{q}} \\ \le \ & C_{n,p,q} \, \frac{\varepsilon^{n} + \eta^n}{\eta^{d(1+\frac{1}{q})}}\,  (1+\vert t -t_0\vert)^{d(1+\frac{1}{q})} \Vert A \Vert \vert X \vert^{1+\frac{1}{p}+\frac{1}{q}}, 
	\end{align*}
	where the second inequality follows from Lemma \ref{resummation3}, the third inequality is discussed above. The last inequality is true since the series converges. 
	As the bound is independent of $\Lambdak$ we may take the limit $k \to \infty$, which exists by Lemma \ref{resummation2}, and get the result. 
\end{proof}
 
\section{Observables in $\mathcal{D}_f$} \label{localtoflocalizedappendix}
Here we show how to generalise our main results to observables $A \in \mathcal{D}_f\supset \Aloc$. 
\begin{lem} \label{localtoflocalized}  
	Let $f \in \mathcal{S}$, $\omega$ a state on $\mathcal{A}$ and $T: D(T) \to \mathcal{A}$ be a densely defined linear operator with $ \Aloc  \subset \mathcal{D}_f \subset D(T) \subset \mathcal{A}$ and $\Vert T\circ (\mathbb{E}_{\LambdaN} -{\mathrm{id}})\Ab{A}\Vert \to 0$ as $N \to \infty$ for all $A \in \mathcal{D}_f$. In particular, $T$ could be an automorphism on $\mathcal{A}$. 
	Assume that there exists a constant $C>0$ and $b \in \mathbb{N}$ satisfying
	\begin{equation}
	\vert \omega(T\Ab{A}) \vert \le C \Vert A \Vert \vert X \vert^b, \quad A \in \mathcal{A}_X,
	\end{equation}
	for any $X \in \mathcal{P}_0(\Gamma)$. Then there exists a constant $C_{b,f}>0$ such that
	\begin{equation}
	\vert \omega(T\Ab{A}) \vert \le C \, C_{b,f} \Vert A \Vert_f, \quad A \in \mathcal{D}_f.
	\end{equation}
\end{lem}
\begin{proof}
	Let $A \in \mathcal{D}_f$ and write for the local approximation
	\begin{equation*}
	A = \lim\limits_{N \to \infty}\sum_{j=0}^{N} A^{(j)} = \lim\limits_{N \to \infty} \mathbb{E}_{\LambdaN}\Ab{A},
	\end{equation*}
	where 
	\begin{align*}
	A^{(0)} = \mathbb{E}_{\Lambda_1}\Ab{A}\quad \text{and} \quad A^{(j)} = (\mathbb{E}_{\Lambda_{j+1}} - \mathbb{E}_{\Lambda_j})\Ab{A} \quad \text{for} \quad j \ge 1.
	\end{align*}
	The limit exists in norm and we also have $\omega(T\circ \mathbb{E}_N\Ab{A})) \to \omega(T\Ab{A})$ as $N \to \infty$. Hence, 
	\begin{align*}
	\vert \omega(T\Ab{A}) \vert &= \lim\limits_{N \to \infty} \vert \omega(T\circ \mathbb{E}_N\Ab{A} ) \vert \le  \lim\limits_{N \to \infty} \sum_{j=0}^{N}\vert \omega(T\Ab{A^{(j)}})\vert \\
	&\le \lim\limits_{N\to \infty} \sum_{j=1}^{N} C \, \Vert (\mathbb{E}_{\Lambda_{j+1}} - \mathbb{E}_{\Lambda_j})\Ab{A}\Vert\, \vert \Lambda_{j+1}\vert^b + C \, \Vert \mathbb{E}_{\Lambda_1}(A) \Vert \, \vert \Lambda_1 \vert^b \\ &\le C \left( 2\ \sum_{j=1}^{\infty}  f(j) \vert \Lambda_{j+1}\vert^b +  \vert \Lambda_1 \vert^b \right) \Vert A \Vert_f\,. \qedhere
	\end{align*}
\end{proof}
In the following Lemma, we show that the maps $\mathcal{K}_j^{\varepsilon, \eta}(t) : \Aloc  \to \mathcal{A}$ are closable and drop the superscripts $\varepsilon$ and $\eta$. 
\begin{lem} \label{closureofderivation}
Let $j \in \mathbb{N}_0$ and $\mathcal{K}_j(t) : \Aloc  \to \mathcal{A}$ be the limit of $\mathcal{K}_j^{\Lambdak}(t):\mathcal{A}_{\Lambdak} \to \mathcal{A}_{\Lambdak}$ (see Lemma \ref{cauchy4}). Then $(\mathcal{K}_j(t), \Aloc )$ is closable for any $t \in I$.
\end{lem}
\begin{proof}
The proof of this lemma rests on two facts: First, the maps $\mathcal{K}_j^{\Lambdak}(t):\mathcal{A}_{\Lambdak} \to \mathcal{A}_{\Lambdak}$ are defined by the condition 
\begin{equation*}
\mathrm{e}^{-\mathrm{i}\sum_{\mu=1}^{n}\varepsilon^{\mu} A^{\Lambdak}_{\mu}(t)} B \mathrm{e}^{+\mathrm{i}\sum_{\mu=1}^{n}\varepsilon^{\mu} A^{\Lambdak}_{\mu}(t)} = \sum_{j=0}^{n} \varepsilon^j \mathcal{K}_j^{\Lambdak}(t)\Ab{B} + \mathcal{O}(\varepsilon^{n+1})
\end{equation*}
for any local observable $B \in \mathcal{A}_{\Lambdak}$, all $t \in I$ and $n \in \mathbb{N}_0$, uniformly in $k \in \mathbb{N}$. Second, $\mathcal{K}_j(t) : \Aloc  \to \mathcal{A}$ is the limit of $\mathcal{K}_j^{\Lambdak}(t):\mathcal{A}_{\Lambdak} \to \mathcal{A}_{\Lambdak}$ by definition. 

For the crucial step we need the following observation: Let $\mathcal{A}$ be a $C^*$-algebra and~$\alpha$ a $^*$-automorphism of $\mathcal{A}$. Then
\begin{equation*}
\alpha\Ab{B^* B} + B^* B \ge \alpha\Ab{B^*} B + B^* \alpha\Ab{B}
\end{equation*}
for any $B \in \mathcal{A}$, since
\begin{equation*}
\alpha\Ab{B^* B}  - \alpha\Ab{B^*} B - B^* \alpha\Ab{B}+ B^* B = (\alpha\Ab{B} - B)^* (\alpha\Ab{B} - B) \ge 0. 
\end{equation*}
As the conjugation with unitaries is a $^*$-automorphism and $\mathcal{K}_j^{\Lambdak}(t)\Ab{B} \to \mathcal{K}_j^{}(t)\Ab{B}$ as $k \to \infty$, we get
\[
 \sum_{j=0}^{n} \varepsilon^j \mathcal{K}_j^{}(t)\Ab{B^* B} + B^* B + \mathcal{O}(\varepsilon^{n+1})  \ge \sum_{j=0}^{n} \varepsilon^j \mathcal{K}_j^{}(t)\Ab{B^*}B + \sum_{j=0}^{n} \varepsilon^j B^* \mathcal{K}_j^{}(t)\Ab{B}+ \mathcal{O}(\varepsilon^{n+1})
\]
for any $B \in \mathcal{A}_{\mathrm{loc}}$, all $t \in I$ and $n \in \mathbb{N}_0$. Since $\mathcal{K}_0(t)$ is the identity on $\mathcal{A}_{\mathrm{loc}}$ (which is trivially closable) and by comparing coefficients, we get
\begin{align*}
 \mathcal{K}_j^{}(t)\Ab{B^* B} 
\ge \mathcal{K}_j^{}(t)\Ab{B^*}B +  B^* \mathcal{K}_j^{}(t)\Ab{B} 
\end{align*}
for all $j \in \mathbb{N}$, $B \in \mathcal{A}_{\mathrm{loc}}$ and $t \in I$. 

Since $\mathbf{1} \in \mathcal{A}_{\mathrm{loc}} $, $B \in\mathcal{A}_{\mathrm{loc}}$ with $B \ge 0 $ implies $B^{1/2} \in \mathcal{A}_{\mathrm{loc}}$ and $ \mathcal{K}_j^{}(t)\Ab{B}^* =  \mathcal{K}_j^{}(t)\Ab{B^*} $ for all $B \in \mathcal{A}_{\mathrm{loc}}$ (by definition), we get by application of Proposition~3.2.22 in \cite{bratteli2012operator} that $\mathcal{K}_j^{}(t)$ is dissipative and therefore closable by Proposition~3.1.15 in \cite{bratteli2012operator}. 
\end{proof}
In the following Lemma, we show that $\overline{\mathcal{K}}_j(t)$ satisfies the assumption of Lemma \ref{localtoflocalized}. 
\begin{lem}  {\rm (Adaption of Lemma 4.12 from \cite{MO})}\label{domainofderivlemma}  \\
Let $j \in \mathbb{N}_0$ and $\overline{\mathcal{K}}_j(t): D(\overline{\mathcal{K}}_j(t)) \to \mathcal{A}$ be the closure of the thermodynamic limit of $\mathcal{K}_j^{\Lambdak}:\mathcal{A}_{\Lambdak} \to \mathcal{A}_{\Lambdak}$. 
Let $f: [0,\infty) \to (0,\infty)$ be a bounded, non-increasing function such that 
\begin{equation*}
\sum_{k=1}^{\infty} (k+1)^{j\cdot d} f(k) < \infty.
\end{equation*}
Then $\mathcal{D}_f \subset D(\overline{\mathcal{K}}_j(t))$ for all $t \in I$ and there is a constant $C_{f} > 0 $ such that
\begin{equation*}
\sup_{t \in I} \Vert \overline{\mathcal{K}}_j(t)\Ab{A}\Vert \ \le \ C_{f}\Vert A \Vert_f 
\end{equation*}
for all $A \in \mathcal{D}_f$. So, using the properties of $\mathcal{D}_f$ (see Appendix B of \cite{MO}), if $g: [0,\infty) \to (0,\infty)$ is a bounded, non-increasing function such that $\lim\limits_{N \to \infty} \frac{g(N)}{f(N)} =0 $, we have in particular
\begin{equation*}
\lim\limits_{N \to \infty} \sup_{t \in I}\Vert  \overline{\mathcal{K}}_j(t) \circ (\mathrm{id} -  \mathbb{E}_{\LambdaN})\Ab{A} \Vert = 0
\end{equation*}
for all $A \in \mathcal{D}_{g}\subset \mathcal{D}_f$.
\end{lem}
\begin{proof}
The proof of this lemma is analogous to the one of Lemma 4.12 in \cite{MO} with minor changes, since we are not dealing with derivations generated by finite range interactions. Recall, that $\mathcal{K}_j(t)\Ab{A}$ (we drop $\overline{\,\cdot\,}$ from now on) is obtained as the limit of a finite sum of multi-commutators with the operators $A_1(t), ... , A_j(t) \in \mathcal{L}_{\mathcal{S}, \infty}$, so each summand has a maximum of $j$ commutators. 
Using Lemma~C.3 from \cite{MT}, we have for $A \in \mathcal{A}_X$ the estimate
\begin{align*}
\Vert \mathcal{K}_j(t)\Ab{A} \Vert &\le C \Vert A \Vert \vert X \vert^j, \quad t \in I,
\end{align*}
where $C>0$ is independent of $A$ and $t$. 
From this, for any $A \in \mathcal{D}_f$ and $N,M \in \mathbb{N}$ with $N>M$, we have
\begin{align}
 &\Vert \mathcal{K}_j(t)\circ (\mathbb{E}_{\LambdaN} -   \mathbb{E}_{\LambdaM})\Ab{A} \Vert \;=\; \left\Vert \sum_{k=M}^{N-1} \mathcal{K}_j(t)\circ (\mathbb{E}_{\Lambda_{k+1}} - \mathbb{E}_{\Lambdak} )\Ab{A}\right\Vert \nonumber\\ &\le  \; C \sum_{k=M}^{N-1} \vert \Lambda_{k+1}\vert^j \Vert (\mathbb{E}_{\Lambda_{k+1}} - \mathbb{E}_{\Lambdak} )\Ab{A}\Vert \; \le \;2C \left(\sum_{k=M}^{N-1} \vert \Lambda_{k+1}\vert^j f(k) \right) \Vert A \Vert_f. \label{derivationcauchy}
\end{align}
This implies, that $(\mathcal{K}_j(t)\circ \mathbb{E}_{\LambdaN} \Ab{A})_{N \in \mathbb{N}}$ with $A \in \mathcal{D}_f$ is a Cauchy sequence in $\mathcal{A}$, hence there exists a limit. Moreover, $\mathbb{E}_{\LambdaN}\Ab{A}$ converges to $A$ in $\Vert \cdot \Vert $. Since $\mathcal{K}_j(t)$ is closed, $A \in \mathcal{D}_f$ belongs to the domain $D(\mathcal{K}_j(t))$ of $\mathcal{K}_j(t)$ and 
\begin{equation*}
\mathcal{K}_j(t)\Ab{A}= \lim\limits_{N \to \infty} \mathcal{K}_j(t)\circ \mathbb{E}_{\LambdaN}\Ab{A}\,.
\end{equation*}
Hence,   $  \mathcal{D}_f \subset D(\mathcal{K}_j(t))$. Similarly to the estimate in \eqref{derivationcauchy}, one obtains
\begin{equation*}
\Vert \mathcal{K}_j(t)\Ab{A} \Vert \le \left(2C  \sum_{k=1}^{\infty} \vert \Lambda_{k+1}\vert^j f(k)  + C \vert \Lambda_1\vert^j \right) \Vert A \Vert_f = C_f \Vert A\Vert_f  
\end{equation*}
for any $A \in \mathcal{D}_f$ by considering   a limit.
Hence, we have proven the claim. 
\end{proof}

  \noindent {\bf Acknowledgment.} J.H.~acknowledges partial financial support by the ERC Advanced Grant ``RMTBeyond" No.~101020331.
  \\[5mm]
	\noindent \textbf{Data Availability.}
	Data sharing is not applicable to this article as no new data were created or analyzed in this study.

\end{document}